\documentclass{article}

\usepackage{array}

\usepackage{arxiv}
\usepackage{comment}        
\usepackage[utf8]{inputenc} 
\usepackage[T1]{fontenc}    
\usepackage{hyperref}       
\usepackage{url}            
\usepackage{booktabs}       
\usepackage{amsfonts}       
\usepackage{bbold}          
\usepackage{amsmath}        
\usepackage{amsthm}         
\usepackage{stmaryrd}       
\usepackage{nicefrac}       
\usepackage{microtype}      
\usepackage{lipsum}		    
\usepackage{graphicx}
\usepackage{natbib}
\usepackage{doi}
\usepackage{amsmath}
\usepackage{mathtools}      
\usepackage{float}
\usepackage[dvipsnames]{xcolor}
\usepackage{enumitem} 
\usepackage{multirow}
\usepackage{breqn} 

\newtheorem{theorem}{Theorem}[section] 

\newtheorem{corollary}{Corollary}[theorem]

\title{Financial Stochastic Models Diffusion: From Risk-Neutral to Real-World Measure}

\author{ {Mohamed Ben~Alaya} \\
	LMRS, UMR 6085 CNRS\\ 
	Université de Rouen Normandie\\
	1 Rue Thomas Becket, 76130 Mont-Saint-Aignan\\
	\texttt{mohamed.ben-alaya@univ-rouen.fr} \\
	\And
{Ahmed ~Kebaier} \\ 
	LAMME, UMR 8071\\
	Université d'Evry\\
	23 Bd. de France, 91037 Évry Cedex\\
	\texttt{ahmed.kebaier@univ-evry.fr} \\
 	\And
{Djibril~Sarr} \\
	LAGA, UMR  7539 CNRS\\
	Université Sorbonne Paris Nord\\
	99 Av. Jean Baptiste Clément, 93430 Villetaneuse\\
	FBH Associés \\
	11 Rue du 4 septembre, 75002 Paris\\
	\texttt{sarr@math.univ-paris13.fr} \\
	\texttt{djibril.sarr@fbh-associes.com} \\
}

\renewcommand{\headeright}{}
\renewcommand{\undertitle}{}
\renewcommand{\shorttitle}{Financial Stochastic Models Diffusion: From Risk-Neutral to Real-World Measure}

\hypersetup{
pdftitle={Financial Stochastic Models Diffusion: From Risk-Neutral to Real-World Measure},
pdfsubject={Spreads Models, Real-World, Financial modeling},
pdfauthor={Djibril SARR},
pdfkeywords={Real-World, Girsanov, Credit spreads},
}

\begin{document}
\maketitle

\begin{abstract}

This research presents a comprehensive framework for transitioning financial diffusion models from the risk-neutral (RN) measure to the real-world (RW) measure, leveraging results from probability theory, specifically Girsanov's theorem. The RN measure, fundamental in derivative pricing, is contrasted with the RW measure, which incorporates risk premiums and better reflects actual market behavior and investor preferences, making it crucial for risk management. We address the challenges of incorporating real-world dynamics into financial models, such as accounting for market premiums, producing realistic term structures of market indicators, and fitting any arbitrarily given market curve. Our framework is designed to be general, applicable to a variety of diffusion models, including those with non-additive noise such as the CIR++ model. 
Through case studies involving Goldman Sachs' 2024 global credit outlook forecasts and the European Banking Authority (EBA) 2023 stress tests, we validate the robustness, practical relevance and applicability of our methodology. This work contributes to the literature by providing a versatile tool for better risk measures and enhancing the realism of financial models under the RW measure. Our model's versatility extends to stress testing and scenario analysis, providing practitioners with a powerful tool to evaluate various what-if scenarios and make well-informed decisions, particularly in pricing and risk management strategies.

\end{abstract}


\section{Introduction}
\label{sec:SOA}
Stochastic financial modeling has historically been done under the risk-neutral (RN) measure $\mathbb{Q}$, which is a fundamental concept in financial mathematics. This is mainly due to the fact that major stochastic models, especially those aiming at the diffusion of the instantaneous interest rate, are primarily used for market-related applications. Indeed, RN modeling is very useful when pricing instruments and when making classic risk computations (e.g., fair-value losses, CVA computation) as it ensures arbitrage-free modeling 
and that discounted prices of instruments (e.g., the discounted option price in the Black-Scholes framework, \cite{black1973pricing}) are martingales. Also, in general, RN modeling has many practical advantages. By assuming that all securities grow at the risk-free rate in the risk-neutral world, pricing derivatives becomes equivalent to finding the expected value of their discounted payoffs under this measure. Many financial models are analytically tractable under the risk-neutral measure, allowing for the derivation of analytical formulas. They are, in general, more straightforward to implement.

The Risk-Neutral (RN) probability measure is often contrasted with the Real-World measure (RW). The latter, often denoted as $\mathbb{P}$, tries to reflect the actual behavior of real financial markets. For instance, \cite{rw_g2++} introduce the RW measure by including the market price of risks in the drift of the instantaneous interest rate described by a one-factor short rate RN-model, making the models more complex but also more realistic. Indeed, unlike the risk-neutral measure where the expected return is the risk-free rate, the real-world measure includes risk premiums. These are the additional returns that investors expect to compensate for the risk they undertake. This makes the real-world measure more reflective of actual investor behavior and market dynamics. 
In a more general definition, that goes beyond interest rate models, we can define the RW measure as the measure that takes into consideration the current observed behavior and/or the future expected behavior of market indicators (interest rates, risk indicators, investor preferences, etc.).
It is becoming increasingly established that due to pricing needs, the focus has historically been on the $\mathbb{Q}$-measure, as we have already explained, it is often mandatory for pricing. However, for applications that require realistic scenarios, such as regulatory compliance, capital requirement calculations, long-term financial planning, risk management, macroeconomic forecasting, and performance scenario analyses, the $\mathbb{P}$-measure is crucial (e.g., \citep{rw_g2++, rw_hull}). It helps in understanding how financial instruments and portfolios might perform under various real-world conditions. For instance, the authors in \cite{rw_hull} point out that the Basel II Framework\footnote{See footnote 240 on page 261 of \href{https://www.bis.org/publ/bcbs128.pdf}{International Convergence of Capital Measurement and Capital Standards}} already admitted that when computing a counterparty's exposure for CVA (Credit Valuation Adjustment), the RW probability should be used instead of the risk-neutral one, but due to complexities arising from RW modeling, they tolerated using RN forecasting models. RW simulation should respect some properties to allow their optimal use. According to \cite{rw_bgm}, in their work regarding RW modeling of the Brace Gatarek Musiela Model (BGM model, also known as the Libor Market Model, LMM, \cite{BGM_base}), RW models should be arbitrage-free, produce realistic dynamics, and be able to match a predetermined curve. These are all properties that we will ensure are verified in our different RW applications. 


Our work aims to provide researchers and practitioners with a general approach that allows one to shift a model from RN to RW using change of measures results from probability theory. The generic framework we developed suits several diffusion models, including ones with non-additive noise like the CIR model \cite{CIR_base}. Moreover, our framework is not solely focused on the market price of risk and is designed to address a broad spectrum of market indicators, including interest rates and credit spreads. This work provides a tool to observe the changes in the entire term structure of a market indicator described by a diffusion model when values are anticipated for a certain maturity. These anticipations can come from regulatory authorities for region-wide stress tests or from internal forecasts or stress values. Numerical tests are performed, demonstrating the applications of our generic framework on regulatory stress scenarios and macroeconomic forecasting.

Despite significant advancements in incorporating RW measures in financial modeling to better reflect actual market dynamics and investor behavior, the literature is not as extensive as other topics in financial modeling. Most, if not all, of the work in the literature dealing with RW diffusion is focused on interest rates. To our knowledge, this work is the first to provide a generic framework intended to be applied to many risk indicators (like credit spreads) and to a wide range of models, including those with non-additive noise (like the CIR++ model).  

\cite{rw_g2++} conducted a study with many similarities to the work we will present in Section \ref{sec:RWmodel}, but focused on zero-coupons and the G2++ model (also called the Gauss2++ or Gaussian two-factor model). The problem is therefore quite similar, but takes advantages of the regularity properties of Gaussian models, especially ones with additive noise. Their paper introduces a framework for calibrating the G2++ model under both RN and RW measures. They introduce a time-dependent market price of risk, whereas previous work considered a constant market price of risk (for instance, \cite{ConstantMarketPriceOfRisk}), which is known to be inefficient when forecasting interest rates (see \cite{BadForecast}). The authors also conduct a comparative analysis between step and linear functions to describe the market price of risk. These time-dependent market price of risk functions allow them to achieve more stable long-term interest rate forecasts.

Complementing this approach, \citep{bruti2010rw} also study interest rate term structure models under the RW measure. In this context, they explore the integration of jump-diffusion processes, capturing both continuous and discrete market uncertainties. The particularity of their work is twofold. First, it is in their choice of the growth optimal portfolio as the numeraire for their pricing, and secondly, the fact that they take an interest in pricing under the RW measure. The authors aim to incorporate real-world trends into the pricing of interest rate derivatives. 

From a numerical point of view, \citep{barker2016rw} tackle the simulation aspect of RW measure estimation. They propose and compare two methodologies: a parametric approach using the Esscher transform and a minimum entropy, non-parametric approach for estimating the RW measure within Monte Carlo simulation frameworks. The authors also apply their work to interest rate models, specifically the 3-factor Hull-White model.


This analysis of the state of the art shows that most of the RW modeling has been performed to accommodate interest rate forecasting. This explains why an important part of RW modeling has, as a practical aim, the calibration of the market price of risk. In our attempt to address the limitations of traditional RN models and improve risk management and forecasting, the focus of our work goes beyond the calibration of the market price of risk. We are mainly interested in allowing RW distributions that can take our forecasts of the future into consideration by replicating any arbitrarily given curve for a specific maturity, and allowing one to observe how the whole term structure evolves as a consequence.

In what follows, we will start by deriving in Section \ref{sec:RWmodel} a generic framework to switch from RN to RW. Then we will show how the framework can be applied. 
In Section \ref{sec:AppCIR}, we apply the model to the RW modeling of credit spreads via the CIR++ intensity model \cite{SarrCreditSpreads}. Finally, Section \ref{sec:rwresults} introduces two numerical results. We conduct two Monte Carlo simulations for credit spreads under the RW measure: the first application is for economic forecasting, and the second is for stress testing. In both examples, we demonstrate how the model can be used to exactly fit an arbitrarily given curve.

\section{General framework and main results}
\label{sec:RWmodel}
\subsection{Aim of this Section}
Financial institutions as well as their regulatory authorities are continuously more demanding when it comes to risk assessment and assets and liabilities management. There is a real need of a better understanding of the future behavior of market risk indicators. The problem of modelling the Real-World (RW) evolution of the term structure of key risk indicators is an important one, yet it has not received much interest relatively to Risk-Neutral (RN) models. Also, most of the research conducted focuses on interest rates. Indeed, most term structure models have been developed in order to price interest rate derivatives, (swaptions, caplets, …). In pricing problems, it is generally necessary to use risk-neutral probabilities, as they guarantee arbitrage free prices. However, Real-World probabilities are necessary to answer questions about the real-world distribution of market indicators in the future. Two main reasons can explain that:

\begin{itemize}
    \item When forecasting, RN models do not take into account the fact that investors demand a risk premium for risky assets. This is well explained in \citep{rw_lmm_neg}. 
    
    \item RN model do not allow one to make \textit{what-if} scenarios where we observe the deformation of the whole term structure of a risk indicators under a constraint on one or more maturities.
\end{itemize}

In this section, we introduce a generic framework, adapted to many diffusion models, to consistently switch from RW to RN. The features we want our RW model to exhibit are driven by the characteristics identified by \cite{rw_bgm}, where the author states that a set of desirable properties for a RW model to verify is:

\begin{enumerate}
    \item Being arbitrage-free; 
    \item Producing realistic term structures;
    \item Being able to match any arbitrarily given curve.
\end{enumerate}

The feature that will receive most of our attention in this Section and the following is the third, as it is the hardest to guarantee while being of high interest in risk management.

The task of expressing a chosen indicator under the RW measure can be very tedious, especially with non-additive noise. This is why, in the theoretical framework below, we start by eliminating non-additive noise through a Lamperti transformation. Our approach is based on expressing the transformed indicator under the RW measure as a function of (i) the expression of the transformed indicator under RN and (ii) a parametric function that needs to be calibrated. Once these relations are found, we can revert the transform to express the indicator under RW as a function of (i) the indicator under RN, (ii) a parametric function, and (iii) other terms that might arise from reverting the transformations. Once this is obtained, through numerical simulations involving the calibration of the parametric function and the simulations of the indicator under the RN measure, we can simulate the chosen indicator under RW, fitting any given value. Furthermore, as we will see in the applications, especially in Section \ref{sec:AppCIR}, transitioning from a model defined in RN to a model in RW will be sufficient to justify the first property. The second property requires numerical simulations, which will be performed in Section \ref{sec:rwresults}. The aforementioned parametric function will be calibrated to allow the real-world diffusion of the indicator to fit any given values for the indicator. 

The indicator we want to fit to given values is very rarely the \textit{initial quantity} that is diffused by the model being considered. For instance, when working with interest rate models, the model's \textit{initial quantity} is the instantaneous rate. However, the quantity that practitioners and researchers are mostly interested in will often be zero-coupon bond rates, zero-coupon bond prices, or swaption prices. In the application we will see in Section \ref{sec:AppCIR}, the CIR++ intensity model describes the diffusion of the default intensity, but our indicators of interest will be: credit spreads cumulative hazard rates.

In the remainder of this section, we outline a general approach to establishing the relation between the \textit{initial quantity} (the one diffused by the considered model) under RW as a function of the three aforementioned quantities (the indicator under RN, a parametric function, and other terms). Then, in the application in Section \ref{sec:AppCIR}, we will consider a practical example using the CIR++ intensity model. We will move from the relations involving the \textit{initial quantity} (default intensity) (see Theorem \ref{theo:rwrn} and equation \eqref{eq:lamstar}) to the relations involving the real indicators of interest (credit spreads and cumulative hazard rates) (see Theorems \ref{theo:spspstar} and \ref{theo:linkLam}).

\subsection{Theoretical framework}
The risk-neutral model is set in a filtered probability space $\left(\Omega, \mathcal{F}, (\mathcal{F}_t)_{t \in[0,T]}, \mathbb{Q} \right)$. Here, $(\mathcal{F}_t)_{t \in[0,T]}$ is a given filtration, and $\mathbb{Q}$ defines the risk-neutral probability measure. Let $(W_t)_{t \in[0,T]}$ be a $\mathcal{F}_t$-Brownian motion under $\mathbb{Q}$. The probability space affiliated with the real-world measure will be denoted by $\mathbb{P}$. Now, let us consider in RN, the process $(Y_t)_{t \in[0,T]}$, which is the solution to
\begin{equation}
\label{eq:eq1}
dY_t = b(Y_t) dt + \sigma(Y_t) dW_t,
\quad Y_0 = y\in \mathbb R,
\end{equation}
where  the functions $b: \mathbb{R} \rightarrow \mathbb{R}$ and $\sigma: \mathbb{R} \rightarrow \mathbb{R}$ are locally Lipschitz continuous with $\frac{1}{\sigma(.)}$ being locally integrable. For $\phi(y) = \int_{y_0}^y \frac{1}{\sigma(x)} dx$, if $\sigma \in \mathcal{C}^1$, then by the Lamperti transform, $X_t = \phi(Y_t)$ satisfies the stochastic differential equation
\begin{equation*}
dX_t = L(X_t) dt + dW_t, \quad
X_0 = \phi(y), 
\end{equation*}
$L(x) = \left(\frac{b}{\sigma} - \frac{\sigma'}{2}\right)(\phi^{-1}(x))$. In the sequel,  while still keeping the relation $X_t = \phi(Y_t)$, we consider the same, more generic framework  as in \cite{Alf2013,BdeKeb24}, and consider
 $(X_t)_{t\ge0}$ the solution to the SDE with additive noise defined on $I=(c, +\infty)$, $c \in [-\infty,+\infty[$:
\begin{align}
\label{eq:lam}
    dX_t&=L(X_t)dt+  \zeta dW_t, \quad t\ge 0,\;X_0=x\in I, \zeta \in \mathbb R,
\end{align}
where the drift coefficient $L$ is assumed to meet the following monotonicity condition:
\begin{equation}
\label{eq:CondEq3}
L: I\longrightarrow \mathbb{R}~~is~~C^2,~~\mbox{such that} ~~\exists~K>0,~~ \forall x,x'\in I, x\leq x', L(x')-L(x)\leq K(x'-x).
\end{equation}

Additionally, for some $d\in I$, we assume that 
\begin{equation}
\label{hyp:H1}
v(x)=\displaystyle \int_d^x \int_d^y \exp\Big(-\frac{2}{\zeta^2}\displaystyle \int_z^y L(\xi) d\xi \Big) dz dy\;\text{ satisfies } \lim\limits_{x \to c^+} v(x) =+\infty.\tag{H1}
\end{equation}

According to Feller's test (see e.g., \cite{karatzas1991stochastic}), the conditions \eqref{eq:CondEq3} and \eqref{hyp:H1} guarantee that the SDE \eqref{eq:lam} has a unique strong solution $(Y_t)_{t \ge 0}$ on $I$, which never reaches the boundaries $c$ and $+\infty$. Under this framework, we introduce the following theorem:

\begin{theorem}
    \label{theo:rwrn}
    Let $(X^*_t)_{0\le t\le T}$ be a process that under $\left ( \Omega, \mathcal{F}, \mathbb{Q} \right )$ is a solution to 
    \begin{equation}
        \label{eq:xstarunderq}
        dX_t^* = \left[- \vartheta [(X_u^*-X_u)  - \alpha_u] + L(X_u) \right] + \zeta dW_t.
    \end{equation}

    Let $(\varphi_t)_{t\in[0,T]}$ be an adapted process defined by
    \begin{equation}
    \label{eq:choice}
        \varphi_t:= \frac{1}{\zeta} \left [-(L(X_t^*)-L(X_t)) - \vartheta(X^*_t-X_t -\alpha_t) \right ],
    \end{equation} 
    
    with, $\vartheta \in \mathbb{R}$ and $\alpha_t$ is a deterministic function.
    
    If     
    \begin{equation}
        \label{eq:Novikov}
        \mathbb E^{\mathbb{Q}}\left[\exp\left(\frac 12 \int_0^T \varphi_u^2du\right)\right]<\infty, \tag{H2}
    \end{equation}

    then $\exists \; \mathbb{P}^{\varphi}$, a new probability measure with the following density 
    \begin{equation*}
        \frac{d\mathbb P^{\varphi}}{d\mathbb Q}\,_{|\mathcal F_T}=\exp\left(-\int_0^T\varphi_udW_u-\frac12 \int_0^T \varphi_u^2du\right), 
    \end{equation*} such that 
    \begin{equation}
    \label{eq:Girs}
        W_t^* =W_t + \int_0^t\varphi_udu, 
    \end{equation}
    is a Brownian motion under $\mathbb{P}^{\varphi}$, under this measure, the process $(X_t)^*$ is a solution of 
    \begin{align}
    \label{eq:RW}
        dX^*_t&=L(X^*_t)dt+  \zeta dW^*_t, \quad t\ge 0,\;X^*_0=x\in I,
    \end{align}
    and verifies the following relation with $(X_t)$
    \begin{equation}\label{RWRN}
    X_t^* = X_t + \vartheta \int^t_s\alpha_u e^{-\vartheta (t-u)} du, \mbox{ for all }  0\le s<t\le T. 
 \end{equation}
\end{theorem}

\begin{proof}
    Condition \eqref{eq:Novikov} is Novikov's condition, which allows the application of Girsanov's theorem. This directly provides the first two results: the density of $\mathbb{P}^{\varphi}$ and the existence and expression of the $\mathbb{P}^{\varphi}$-Brownian motion $(W^*_t)_{t \in [0,T]}$.
    
    Then we combine equations \eqref{eq:choice} and \eqref{eq:xstarunderq} to write the dynamic of $(X^*_t)_{t\in[0,T]}$ under $\mathbb{Q}$ as a function of $\varphi_t$, which is:
    \begin{equation*}
        dX_t^* = [L(X_t^*)+ \zeta \varphi_t]dt + \zeta dW_t.
    \end{equation*}
    Plugging \eqref{eq:Girs} in the above equation easily gives the dynamic of $(X^*_t)_{t\in[0,T]}$ under $\mathbb{P}^{\varphi}$, $dX^*_t=L(X^*_t)dt+  \zeta dW^*_t$.
    
    We can then express $(X_t)_{t\in[0,T]}$ under $\mathbb{P}^{\varphi}$ which is $dX_t = [L(X_t) - \zeta \varphi_t]dt + \zeta dW^*_t$. We hence have: 
    \begin{equation*}
        \begin{aligned}
            dX^*_t &= L(X^*_t)dt + \zeta dW^*_t \\
            dX_t &= [L(X_t) - \zeta \varphi_t]dt + \zeta dW^*_t,
        \end{aligned}
    \end{equation*}
    and we can write:
    \begin{equation*}
        \begin{aligned}
            d(X_t^*-X_t) &= \left[ L(X_t^*) - L(X_t) + \zeta_u \right]dt \\
            \text{ which implies } d(X_t^*-X_t) &= -\vartheta (X_t^* - X_t)dt + \vartheta \alpha_t dt \\
            \text{ which gives } e^{-\vartheta t}d[e^{\vartheta t}(X_t^*-X_t)] &= \vartheta \alpha_t dt, \\
        \end{aligned}
    \end{equation*}

    Hence the final result of the theorem:
    $$
        X_t^* = X_t + \vartheta \int^t_s\alpha_u e^{-\vartheta (t-u)} du.
    $$
\end{proof}

The probability measure $\mathbb{P}^{\varphi}$ is a real-world probability, and $(X^*_t)_{t \in [0,T]}$ is a RW process. In what follows, any process marked with an asterisk, $[.]^*$, denotes a RW process (that might be expressed under the risk-neutral or the real-world probability measure). The variable $\vartheta$ could be chosen to simplify the calculations as much as possible (see Section \ref{sec:AppCIR} for a practical illustration of the choice of $\vartheta$).

This leads to the following theoretical result.
\begin{corollary}
    Under the conditions of theorem \ref{theo:rwrn}, we have
    \begin{equation}
        \label{eq:linear}
        Y_t^* = \phi^{-1} \left ( \phi(Y_t) + \vartheta \int^t_s\alpha_u e^{-\vartheta (t-u)} du \right ).
    \end{equation} 
\end{corollary}

From a practical point of view, we require $Y_t^*$ in the following form:
\begin{equation}
    \label{eq:linear2}
    Y_t^* = Y_t + R(\phi(Y_t), \vartheta, \alpha_t, t),
\end{equation}
for some function $R(\dots)$ that can be as complex as $\phi$ requires it to be, see for example, section below. 

At this stage, we have, as we intended, expressed the indicator under RW as a function of the indicator under RN and other terms. Depending on the type of function $\phi$, this expression will be used to calibrate the parametric function $\alpha_t$. We introduce in the next sections one example with a square root function arising in CIR-type models.



\section{Application to CIR++ intensity model for credit spreads}
\label{sec:AppCIR}

\subsection{Credit Spreads modeling under the Risk-Neutral measure}
To apply the generic framework we introduced in Section \ref{sec:RWmodel} to modeling credit spreads, we use the model and the framework developed in \cite{SarrCreditSpreads}. As stated in Section \ref{sec:SOA}, to our knowledge, this is the first work that takes an interest in the RW diffusion of credit spreads. 

The risk-neutral model is set in a filtered probability space $\left(\Omega, \mathcal{F}, (\mathcal{F}_t)_{t \in[0,T]}, \mathbb{Q} \right)$. Here, $(\mathcal{F}_t)_{t \in[0,T]}$ is the natural filtration of the default-less market, and $Q(.)$ is the probability function. Let $r(t) = r_t$ be the instantaneous risk-free rate of the market and $\lambda(t)$ the default intensity. The default intensity process $(\lambda_t)_{t \in[0,T]}$ and the interest rate process $(r_t)_{t \in[0,T]}$ are $\mathcal{F}_t$-measurable. Also, let the default intensity (or hazard rate) $\lambda(t)$ satisfy:
\begin{equation}
    \label{eq:di}
    \lambda(t) dt = Q\left(\tau \in [t, t + dt] \mid \tau > t, \mathcal{F}_t\right),
\end{equation}
and be represented by the following dynamics\footnote{For convenience, in what follows, we identically write $y_t$ or $y(t)$ as well as $\lambda(t)$ or $\lambda_t$.}:
\begin{equation}
    \label{eq:cirrnp}
    \begin{cases}
        \lambda(t) = y(t) + \psi(t), \\
        dy(t) = \kappa\left(\theta - y_t\right) dt + \sigma \sqrt{y_t} dW_t,
    \end{cases}
\end{equation}

where $(W_t)_{t \in [0,T]}$ is a $\mathbb{Q}$-Brownian motion, and $\kappa, \theta, \sigma > 0$. The Feller condition, $2\kappa \theta \geq \sigma^2$, and $y(0) = y_0 > 0$ guarantee $y_t > 0$ under $\mathbb{Q}$. $\psi(t)$ is a deterministic function. Under this framework, the credit spread of maturity $T$ at time $t$ is given by:
\begin{equation}
\label{eq:spfinal}
    \begin{dcases}
        \operatorname{Sp}(t, T) =-\frac{1}{T-t} \ln \left[\delta+(1-\delta) \frac{S^m(0,T)}{S^m(0,t)}\frac{A(0,t)}{A(0,T)}\frac{e^{-B(0,t)y_0}}{e^{-B(0,T)y_0}}A(t,T)e^{-B(t,T)[\lambda(t)-\psi(t)]} \right]\\
        \psi(t) = \lambda^m(t)+D(t)-y_0E(t),
    \end{dcases}
\end{equation}
where, 
\begin{equation}
    \label{eq:AandB}
    \begin{aligned}
        A(t,T) &= \left( \frac{2he^{\frac{1}{2}(\kappa+h)(T-t)}}{2h+(\kappa+h)\left( e^{(T-t)h} - 1 \right )} \right)^{\frac{2\kappa\theta}{\sigma^2}}\text {, }
        B(t,T) = \frac{2\left( e^{(T-t)h} - 1 \right )}{2h+(\kappa+h)\left( e^{(T-t)h} - 1 \right )}\\
        h &= \sqrt{\kappa^2+2\sigma^2},\\
        D(t) &= \frac{d}{dt}ln(A(0,t))\text{, }E(t) = \frac{d}{dt}B(0,t), 
    \end{aligned}
\end{equation}
and $\delta \in [0,1]$ is the recovery rate. An extensive proof and all the related elements can be found in \cite{SarrCreditSpreads}.

Let us now derive our model's shift from RN to RW. The probability space affiliated with the real-world will again be denoted for brevity by $\mathbb{P}$, instead of $\mathbb{P}^{\varphi}$. We introduce $(W_t^*)_{t \in [0,T]}$, a Brownian motion under the RW probability measure $\mathbb{P}$. We will apply the results from Theorem \ref{theo:rwrn} and then justify how we obtain the three features mentioned earlier from our model.
\begin{equation}
    \label{eq:cir}
    \begin{cases}
        & \lambda(t) = y(t) + \psi(t)  \\
        & dy(t) = \kappa\left(\theta-y_t\right)dt + \sigma\sqrt{y_t}dW_t,
    \end{cases} 
\end{equation}
with $2\kappa \theta \geq \sigma^2$ and $y(0) = y_0>0$. $\psi(t)$ is the deterministic fitting function (for the survival probabilities).

Similarly in RW (real-world), under the measure $\mathbb{P}$, there exists a CIR++ diffusion written:
\begin{equation}
    \label{eq:cirrw}
    \begin{cases}
        & \lambda^*(t) = y^*(t) + \psi(t)  \\
        & dy^*(t) = \kappa\left(\theta-y^*_t\right)dt + \sigma\sqrt{y^*_t}dW^*_t,
    \end{cases} 
\end{equation}
with  $y^*(0) = y^*_0>0$.

Using a Lamperti transform, we re-write for $x_t = \sqrt{y_t}$: 

\begin{equation}
    \label{eq:lamp}
    dx_t = \frac{1}{2}\left[\left(\kappa\theta - \frac{1}{4}\sigma^2\right)\frac{1}{x_t} - \kappa x_t \right]dt + \frac{1}{2}\sigma dW_t.
\end{equation}

Therefore regarding the notations in Theorem \ref{theo:rwrn}, we make the following identifications:
\begin{itemize}
    \item $L(x) = \frac{1}{2}\left[\left(\kappa\theta - \frac{1}{4}\sigma^2\right)\frac{1}{x} - \kappa x \right]$; 
    \item $\zeta = \frac{1}{2}\sigma$;
    \item $\phi(x)=\sqrt{x}.$
\end{itemize}

To apply Theorem \ref{theo:rwrn}, we need to very condition \eqref{eq:Novikov}.

\subsection{Condition \ref{eq:Novikov}: Novikov criterion}\mbox{}\\
Per condition \eqref{eq:Novikov}, we need to verify that our function $\varphi_t$ respects Novikov's criterion. Equation \eqref{eq:choice} indicates that the function $\varphi_t$ is given by:
\begin{equation*}
    \varphi_t= \frac{1}{\zeta} \left [-(L(x_t^*)-L(x_t)) - \vartheta(x^*_t-x_t -\alpha_t) \right ],
\end{equation*} 
which, for the identifications we made earlier corresponds to:
$$
    \varphi_t = \frac{2}{\sigma} \left[ - \frac{1}{2} \left( \kappa \theta - \frac{1}{4} \sigma^2 \right) \left( \frac{1}{x_t^*} - \frac{1}{x_t} \right) + \frac{1}{2} \kappa (x_t^* - x_t) - \vartheta (x_t^* - x_t - \alpha_t) \right].
$$

We stated earlier that the parameter $\vartheta \in \mathbb{R}$ should be chosen to simplify as much as possible the calculations. A good choice, seems to be $\vartheta = \frac{1}{2}\kappa$ as it gives:
\begin{equation*}
    \label{eq:phiparticulier}
    \varphi_t = -\frac{1}{\sigma} \left[ \left( \kappa \theta - \frac{1}{4} \sigma^2 \right) \left( \frac{1}{x_t^*} - \frac{1}{x_t} \right) - \kappa \alpha_t \right],
\end{equation*}

$\alpha_u$ is the deterministic function that will shift our RN diffusion drift to allow us to fit RW values. The function $\alpha_u$ could be for instance step-wise constant. 
We now need to show that our function $\varphi_t$ respects Novikov's condition, that is $\mathbb{E}^{\mathbb{Q}}\left[e^{\frac{1}{2}\int_0^T\varphi_u^2du}\right]<\infty$. 

Let $\gamma=\frac{1}{\sigma}\left(\kappa \theta - \frac{1}{4} \sigma^2\right)$. We need to show that $\mathbb{E}^{\mathbb{Q}}\left[e^{\frac{1}{2}\int_0^T\left[\gamma\left(\frac{1}{x_u^*}-\frac{1}{x_u}\right)\right]^2 du}\right]<\infty$.

\begin{equation*}
     \mathbb{E}^{\mathbb{Q}}\left[\operatorname{exp}\left({\frac{1}{2}\int_0^T\left(\gamma\left(\frac{1}{x_u^*}-\frac{1}{x_u}\right)\right)^2 du}\right)\right] \leq  \mathbb{E}^{\mathbb{Q}}\left[\operatorname{exp}\left({\frac{\gamma^2}{2}\int_0^T\left(\frac{2}{(x_u^*)^2}+\frac{2}{(x_u)^2}\right) du}\right)\right].
\end{equation*}

Remark that the dynamic of $x_t$ under $\mathbb{P}$ is:
\begin{equation*}
    dx_t = \frac{1}{2}\left[\left(\kappa\theta - \frac{1}{4}\sigma^2\right)\frac{1}{x_t} - \kappa x_t -\sigma \varphi_t \right]dt + \frac{1}{2}\sigma dW_t^*.
\end{equation*}
Meanwhile, the dynamic of $x_t^*$ under $\mathbb{P}$ is: 
\begin{equation*}
    dx_t^* = \frac{1}{2}\left[\left(\kappa\theta - \frac{1}{4}\sigma^2\right)\frac{1}{x_t^*} - \kappa x_t^* \right]dt + \frac{1}{2}\sigma dW_t^*.
\end{equation*}

From these two expressions, the comparison theorem gives us that if $-\sigma \varphi_t < 0$, then $x_t^* > x_t$. The existence, expression, and sign of $\varphi_t$ will be justified later on. Therefore, for a suitable (regarding Girsanov) function $\varphi_t > 0$, we can write:
\begin{equation*}
    \begin{aligned}
         \mathbb{E}^{\mathbb{Q}}\left[\operatorname{exp}\left({\frac{1}{2}\int_0^T\left(\gamma\left(\frac{1}{x_u}-\frac{1}{x_u^*}\right)\right)^2 du}\right)\right]&\leq \mathbb{E}^{\mathbb{Q}}\left[\operatorname{exp}\left({{2\gamma^2}\int_0^T\frac{1}{(x_u)^2} du}\right)\right].\\
    \end{aligned}    
\end{equation*}

Since $(x_u)^2=y_u$, in practical terms, the problem consists of defining the domain of convergence $\mathcal{D}_t$ of the characteristic function of the integral of a solution of \eqref{eq:cir}.

\begin{theorem}
    \label{prop}
    Let $\hat{y}_t$ be the solution of the CIR equation defined in \eqref{eq:cir}, the characteristic function of $\int_0^T\hat{y}_tdt$ is well-defined on  $D_t$ defined by:
    
    \begin{equation*}
        D_t = \left\{\mu \in \mathbb{R}, \; \mathbb{E}\left[{\exp}\left(\mu\int_0^T\frac{1}{\hat{y}_u} du\right)\right]<\infty\right\} = \left\{ \mu, \mu > -\frac{1}{2\sigma^2} \left(\kappa \theta - \frac{1}{2}\sigma^2 \right)^2 \right\}.
    \end{equation*}
\end{theorem}

\begin{proof}
    The expression of the Laplace transform is provided in \cite{BAK}.
    
    \begin{equation}
    \label{eq:charc}
    \begin{aligned}
    \mathbb{E}\left(e^{\mu \int_0^T \frac{d u}{y_u}}\right)= & \frac{\Gamma\left(c+\frac{\nu}{2}+\frac{1}{2}\right)}{\Gamma(\nu+1)} \frac{1}{(y_0)^c \delta^c} \beta^{\frac{\nu}{2}+\frac{1}{2}} \\
    & \times \exp \left(\frac{b}{2 \sigma}\left[\kappa\theta T-\frac{2 y_0}{e^{\kappa T}-1}\right]\right){ }_1 F_1\left(c+\frac{\nu}{2}+\frac{1}{2}, \nu+1, \beta\right),
    \end{aligned}
    \end{equation}
    
    where $\beta=\frac{2\kappa y_0}{\sigma^2\left(e^{\kappa T}-1\right)}, \; c=\frac{\kappa\theta}{\sigma^2}, \; \delta=\frac{2\kappa e^{\kappa T}}{\sigma^2\left(e^{\kappa T}-1\right)}$ and $\nu=\frac{2}{\sigma^2} \sqrt{(\kappa\theta-\frac{1}{2}\sigma^2)^2+2 \mu \sigma^2}$. 
    
    Following a similar approach to the proof of Proposition 1.2.4 in \cite{ALFONSI}, we give the set of convergence of $D_t$. $\left\{\mu \in \mathbb{R}, \mathbb{E}\left[{\exp}\left(\mu\int_0^T\frac{1}{y_u} du\right)\right]<\infty\right\}$ obviously contains at least $\mathbb{R_{-}}$. Also, the characteristic function is analytic on the inside of its set of convergence (e.g. \cite{Widder}). Then $\Gamma$ is analytic on $\mathbb{R} \backslash \{-k, k \in \mathbb{N}^* \}$ (see for example \cite{Handbook}). From \cite{Handbook} we also get that $1/\Gamma$ is entire and therefore analytic in $\mathbb{R} \backslash \{-k, k \in \mathbb{N}^* \}$. We bring to the attention of the reader that even if $\kappa, \; \theta \in \mathbb{R}$, $c>0$ since $2\kappa\theta>\sigma^2$. Then, we get from chapter 6 of \cite{transcen} that the ${ }_1F_1$ function is analytic.   
    Thus, the right hand side of \eqref{eq:charc} is analytic as soon as $\nu$ is well defined in $\mathbb{R}$, which requires $\mu > -\frac{1}{2\sigma^2} \left(\kappa \theta - \frac{1}{2}\sigma^2 \right)^2$ as stated in theorem \ref{prop}.   
    
    The right hand side is analytic in $\mathbb{R} \cap \left\{ \mu, \mu > -\frac{1}{2\sigma^2} \left(\kappa \theta - \frac{1}{2}\sigma^2 \right)^2 \right\}$. We know from the identity theorem in complex analysis (see the first chapter of \cite{identity} for more details) that since the left hand side and the right hand side are both analytic and coincide in $\mathbb{R_{-}} \cap \left\{ \mu, \mu > -\frac{1}{2\sigma^2} \left(\kappa \theta - \frac{1}{2}\sigma^2 \right)^2 \right\}$, they must coincide in $\mathbb{R} \cap \left\{ \mu, \mu > -\frac{1}{2\sigma^2} \left(\kappa \theta - \frac{1}{2}\sigma^2 \right)^2 \right\}$ as well. Hence
    $\left\{\mu \in \mathbb{R}, \mathbb{E}\left[\rm{exp}\left(\mu\int_0^T\frac{1}{y_u} du\right)\right]<\infty\right\} = \left\{ \mu, \mu > -\frac{1}{2\sigma^2} \left(\kappa \theta - \frac{1}{2}\sigma^2 \right)^2 \right\}$ and the expression of $\mathbb{E}\left[\operatorname{exp}\left( \mu \int_0^T\frac{1}{y_u} du\right)\right]$ is given by equation \eqref{eq:charc}. 
\end{proof}

Theorem \ref{prop} justifies that $\varphi_t = -\frac{1}{\sigma}\left(\kappa\theta + \frac{1}{4}\sigma^2\right)\left(\frac{1}{x_t^*}-\frac{1}{x_t}\right)+\kappa\frac{1}{\sigma}\alpha_t$ verifies Novikov's condition. Therefore, we can apply Theorem \ref{theo:rwrn} and write:
\begin{equation}
    \label{eq:xstar}
    x_t^*= x_t + \frac{1}{2} \kappa\int_s^t\alpha_u{e^{-\frac{1}{2} \kappa(t-u)}du}. 
\end{equation}

\subsection{Application to credit spreads and cumulative hazard rate}
We now need to find the relation between the credit spreads term structure in RN and RW. In what follows, we note $f(t)= \frac{1}{2} \kappa\int_s^t\alpha_u{e^{-\frac{1}{2} \kappa(t-u)}du}$. We can naturally draw the relation between the real-world and risk-neutral default intensities. We recall that $x_t=\sqrt{y_t}$ and $\lambda_t=y_t+\psi_t$
\begin{equation}
    \label{eq:lamstar}
    \lambda_t^*= \lambda_t + f_t^2+2f_t\sqrt{y_t}.
\end{equation}

We have so far described and justified, via $\sqrt{y_t}$ (which is $x_t$), a link between $\lambda_t$ and $\lambda_t^*$. Theorem \ref{theo:spspstar} gives the relation between credit spreads in RN and RW.

\begin{theorem} 
    \label{theo:spspstar}
    The real-world term structure of credit spreads, $\operatorname{Sp}^*(t,T)$, defined in the space $\left(\Omega, \mathcal{F}, (\mathcal{F}_t)_{t \in[0,T]}, \mathbb{P} \right)$, is related to the risk-neutral term structure of credit spreads, $\operatorname{Sp}(t,T)$, defined by equation \eqref{eq:spfinal} in the space $\left(\Omega, \mathcal{F}, (\mathcal{F}_t)_{t \in[0,T]}, \mathbb{Q} \right)$, by the following equation:
    \begin{equation}
        \label{eq:spspstar}
            e^{-(T-t)\operatorname{Sp}^*(t,T)} = e^{-B(t,T)F(t,\sqrt{y_t})}\left[ -\delta \left( 1 - e^{B(t,T)F(t,\sqrt{y_t})} \right) + e^{ -(T-t)\operatorname{Sp}(t,T)}\right], 
    \end{equation}
    where $F(t, \sqrt{y_t}):=f_t^2+2f_t\sqrt{y_t}$. 
\end{theorem}

\begin{proof}
    We set in this proof for notation simplicity $M(t, T)=\frac{S^m(0,T)}{S^m(0,t)}\frac{A(0,t)}{A(0,T)}\frac{e^{-B(0,t)y_0}}{e^{-B(0,T)y_0}}A(t,T)$ and $F(t, x_t)=f_t^2+2f_tx_t$. We also have shown $x_t^*=x_t+f_t$ with $f_t=\frac{1}{2} \kappa\int_s^t\alpha_u{e^{-\frac{1}{2} \kappa(t-u)}du}$.
    Since, by construction, equation \eqref{eq:spfinal} can be similarly written for the credit spreads under RW, we have:
  \begin{equation*}
        \begin{aligned}
             \operatorname{Sp}^*(t, T) &=-\frac{1}{T-t} \ln \left[\delta+(1-\delta) M(t,T)e^{-B(t,T)y_t^*} \right]\\
             y_t &= x_t^2 \text{ this implies } y_t^*=y_t+f_t^2+2f_tx_t\\
             \text{ this gives } e^{-(T-t)Sp^*(t,T)} &= \delta + (1-\delta)M(t,T)e^{-B(t,T)y_t}e^{-B(t,T)(f_t^2+2f_tx_t)}\\
             &=e^{-B(t,T)F(t,x_t)}\left[ \delta e^{B(t,T)F(t,x_t)} + \delta - \delta + (1-\delta)M(t,T)e^{-B(t,T)y_t} \right]\\
             &=e^{-B(t,T)F(t,x_t)}\left[ -\delta \left( 1 - e^{B(t,T)F(t,x_t)} \right) + \underbrace{\left(\delta+(1-\delta) M(t,T)e^{-B(t,T)y_t}\right)}_{\operatorname{exp}\left( -(T-t)Sp(t,T) \right)} \right]\\
             e^{-(T-t)Sp^*(t,T)} &= e^{-B(t,T)F(t,x_t)}\left[ -\delta \left( 1 - e^{B(t,T)F(t,x_t)} \right) + e^{ -(T-t)Sp(t,T)}\right].
        \end{aligned}
    \end{equation*}
\end{proof}

So far, we can justify two of the features that we wanted our model to hold. 
\begin{itemize}
    \item \textbf{Arbitrage-free}: The easiest way to justify the arbitrage-free property of this RW model is to use Theorem 2.7 of \citep{arbitragefree}. Since we start from a risk-neutral measure, we ensure that the model is arbitrage-free. This argument is used, for instance, by \citep{rw_bgm}.
    
    \item \textbf{Realistic Term structures}: The natural way to justify that our model produces a realistic term structure of credit spreads is to observe the said curves. This is what Figures \ref{fig:RW_2} and \ref{fig:RW_4} show. We can, however, have an intuitive hint that this property is verified. Let us consider equation \eqref{eq:spspstar}. When considering the trivial case where $\alpha(u)$ is constant at 0, we indeed have the equivalence between $\operatorname{Sp}(.,.)$ and $\operatorname{Sp}^*(.,.)$. Therefore, the curves under RW are realistic as soon as RN curves are, which has already been ensured in \cite{SarrCreditSpreads}. Then, when $(t,T)$ is fixed, $e^{-(T-t)\operatorname{Sp}^*(t,T)}$ is a linearly shifted $e^{-(T-t)\operatorname{Sp}(t,T)}$. These comments on the realistic curve property are also intended to help having a better understanding and interpret equation \eqref{eq:spspstar}.
\end{itemize}

Further developments are, however, needed for the third property, which is the ability to fit any arbitrarily chosen curve. We have given, through equation \eqref{eq:spspstar}, the link between the credit spread at time $t$ and maturity $T$ under RN and RW probability measures. The model we consider fits the initial market survival probabilities or, identically, the initial market cumulative hazard rate curve. The cumulative hazard rate at time $t$ and for the time horizon $T$ is related to the survival probability, $S(t,T)$, by the following equation:
\begin{equation}
    \label{eq:cumhr}
    S(t,T)=e^{-\Lambda(t, T)}.
\end{equation}

We also recall that \cite{SarrCreditSpreads} gives the expression of survival probabilities in our CIR++ intensity context:
\begin{equation}
    \label{eq:survivprob}
    S(t,T)=\frac{S^m(0,T)}{S^m(0,t)}\frac{A(0,t)}{A(0,T)}\frac{e^{-B(0,t)y_0}}{e^{-B(0,T)y_0}}A(t,T)e^{-B(t,T)[\lambda(t)-\psi(t)]}.
\end{equation} 

In what follows, we derive the model to be able to fit any arbitrarily given market cumulative hazard rate curve. This feature is arguably the most convenient and of highest practical use. It will allow researchers practitioners to observe the full credit spread curve deformation as one maturity is set to any arbitrarily chosen value. This provides an efficient way to conduct regulatory or internal stress scenarios. It also provides analysts with a tool to better evaluate the behavior (with respect to credit spreads) of any product they are willing to acquire or sell, and therefore have a better understanding of its forthcoming valuation's variation.

To achieve this feature, we calibrate the function $\alpha(t)$ to ensure that the expectation of $\Lambda^*(t,T)$ in the real-world space corresponds to cumulative hazard rate anticipation / stress values / forecasts, $\vec{c}$. In other terms, we are willing to obtain an equation to solve for $\alpha(u)$ that would guarantee for a chosen maturity $\widetilde{T}$ and for a set of N dates $\{t_i\}_{i\in \llbracket1, N\rrbracket}$ , $\mathbb{E}^{\mathbb{P}}\Lambda^*(t_i, \widetilde{T})=c_i,\; \forall i\in \llbracket1, N\rrbracket$. Theorem \ref{theo:linkLam} gives the relation between the cumulative hazard rate under RN and RW.

\begin{theorem}
    \label{theo:linkLam}
    The real-world cumulative hazard rate, $\Lambda^*(t,T)$, defined in the space $\left(\Omega, \mathcal{F}, (\mathcal{F}_t)_{t \in[0,T]}, \mathbb{P} \right)$, is related to the risk-neutral cumulative hazard rate, $\Lambda(t,T)$, defined by equations \eqref{eq:cumhr} and \eqref{eq:survivprob} in the space $\left(\Omega, \mathcal{F}, (\mathcal{F}_t)_{t \in[0,T]}, \mathbb{Q} \right)$, by the following linear equation:
    \begin{equation}
        \label{eq:LamLamstar}
        \Lambda^*(t,T) = \Lambda(t,T) + B(t,T)F(t, \sqrt{y_t}).
    \end{equation} 
\end{theorem}

\begin{proof}
Equation \eqref{eq:survivprob} similarly describes the survival probability under RW:
\begin{equation*}
    S^*(t,T)=\frac{S^m(0,T)}{S^m(0,t)}\frac{A(0,t)}{A(0,T)}\frac{e^{-B(0,t)y_0}}{e^{-B(0,T)y_0}}A(t,T)e^{-B(t,T)[\lambda^*(t)-\psi(t)]},
\end{equation*}
and we still have the relation $-\operatorname{ln}\left[S^*(t,T)\right]=\Lambda^*(t,T)$.
These two equations lead to: 

\begin{equation*}
    -\operatorname{ln}\left[S^*(t,T)\right]=\Lambda^*(t,T)=-\operatorname{ln}\left[\frac{S^m(0,T)}{S^m(0,t)}\frac{A(0,t)}{A(0,T)}\frac{e^{-B(0,t)y_0}}{e^{-B(0,T)y_0}}A(t,T)\right]+B(t,T)[\lambda^*_t-\psi_t].
\end{equation*}
Since $\lambda^*(t)-\psi(t)=y^*(t)$ and $y^*_t=(x^*_t)^2=(x_t+f_t)^2$, we finally write:
\begin{equation*}
    \begin{aligned}
        \Lambda^*(t,T)&=\underbrace{-\operatorname{ln}\left[\frac{S^m(0,T)}{S^m(0,t)}\frac{A(0,t)}{A(0,T)}\frac{e^{-B(0,t)y_0}}{e^{-B(0,T)y_0}}A(t,T)\right]+B(t,T)y_t}_{\Lambda(t,T)}+ \\ 
        & B(t,T) \underbrace{[f_t^2+2f_tx_t]}_{F(t,x_t)}.\\
        \text{ From which we finally have } \Lambda^*(t,T) &= \Lambda(t,T) + B(t,T)F(t, \sqrt{y_t}).
    \end{aligned}
\end{equation*}
\end{proof}

Taking the expectations under the RW probability measure $\mathbb{P}$ on equation \eqref{eq:LamLamstar}, it is straightforward to write: 
\begin{equation*}
    \label{eq:LamLamstarExpec}
   \mathbb{E}^{\mathbb{P}}\Lambda^*(t,T) = \mathbb{E}^{\mathbb{P}}\Lambda(t,T) + B(t,T)[f_t^2+2f_t\mathbb{E}^{\mathbb{P}}\sqrt{y_t}],
\end{equation*}

Then we can derive the equation to solve $\forall t_i \in \{t_i\}_{i \in \llbracket1,N \rrbracket}$ for maturity $\widetilde{T}$:

\begin{equation*}
       \mathbb{E}^{\mathbb{P}}\Lambda^*(t_i,\widetilde{T}) = \mathbb{E}^{\mathbb{P}}\Lambda(t_i,\widetilde{T}) + B(t_i,\widetilde{T})\left[ f^2(t_i)+2f(t_i)\mathbb{E}^{\mathbb{P}}\sqrt{y(t_i)} \right],
\end{equation*}

which implies 
\begin{equation}
    \label{eq:tosolve}
    \left[c_i - \mathbb{E}^{\mathbb{P}}\Lambda(t_i,\widetilde{T})\right]\frac{1}{B(t_i,\widetilde{T})}=f^2(t_i)+2f(t_i)\mathbb{E}^{\mathbb{P}}\sqrt{y(t_i)},
\end{equation}
with $f(t)=\frac{1}{2} \kappa\int_s^t\alpha_u{e^{-\frac{1}{2} \kappa(t-u)}du}$.

We are solving \eqref{eq:tosolve} for $\alpha(t)$. We can choose many types of functions for $\alpha(t)$. \citep{rw_g2++} analyze the impact of choosing a constant, step, or linear function. In this paper, we choose step functions as they are an easy way to obtain time-varying parameters. Other possibilities will not be discussed here. Therefore, $\alpha(u)$ is a step function that will have exactly $N$ steps $\alpha^i$, $i \in \llbracket 1, N \rrbracket$ (the function will remain constant at its last value after the last constraint value).
$$
    \alpha(t)=\sum_{i=1}^{N-1}\mathbb{1}_{t \in [t_i,t_{i+1}]} \alpha^i + \alpha^N\mathbb{1}_{t>t_N}.
$$

It comes the following expression for $f(t)$:

\begin{equation}
    \label{eq:ft}
    f(t) = 
    \begin{dcases}
        e^{-\frac{1}{2}\kappa t} \sum_{i=1}^m\alpha^i\left( e^{\frac{1}{2}\kappa t_{i+1}} - e^{\frac{1}{2}\kappa t_{i}} \right) \text{ if } t\leq t_N \text{ and $m$ s.t. } t \in [0, t_m] \\
        e^{-\frac{1}{2}\kappa t} \sum_{i=1}^{N-1}\alpha^i\left( e^{\frac{1}{2}\kappa t_{i+1}} - e^{\frac{1}{2}\kappa t_{i}} \right) + e^{-\frac{1}{2}\kappa t}\alpha^N\left( e^{\frac{1}{2}\kappa t} - e^{\frac{1}{2}\kappa t_{N}} \right) \text{ if } t> t_N.
    \end{dcases}
\end{equation}

The proof is given in appendix \ref{sec:ftproof}.

\paragraph{$\bullet$  Solving for $\alpha(t)$} \mbox{} \\
Let us now explain how each element of equation \eqref{eq:tosolve} is obtained: 

\begin{itemize}
    \item The set of scalars $\{c_i\}_{i \in \llbracket 1, N \rrbracket}$ are the expert-wise chosen targets of the cumulative hazard rate. They can be forecasts, stress values and so on. Each $c_i$ is a target cumulative hazard rate value for the target date $t_i$.

    \item The expectation $\mathbb{E}^{\mathbb{P}}\Lambda(t_i,\widetilde{T})$ is the expectation of the risk-neutral cumulative hazard rate for the horizon $\widetilde{T}$ rate under the the real-world probability. It is obtained for each $t_i$ via Monte-Carlo simulation.

    \item The scalar $B(t_i,\widetilde{T})$ is obtained via equation \eqref{eq:AandB}. 

    \item The scalar $f(t_i)$ is obtained via equation \eqref{eq:ft}.

    \item The expectation $\mathbb{E}^{\mathbb{P}}\sqrt{y(t_i)}$ is also simulated in the RW Monte-Carlo process using $y(t_i)$  that is also required to obtain $\mathbb{E}^{\mathbb{P}}\Lambda(t_i,\widetilde{T})$.
    
\end{itemize}

\paragraph{$\bullet$  Credit spreads anticipation}\mbox{}\\
It is clear that practitioners willing to use a RW credit spreads model will focus on credit spreads rather than cumulative hazard rates. Similarly, regulatory stress tests are concerned with credit spreads and not cumulative hazard rates. In the previous sections, we presented the link between RW and RN credit spreads (equation \eqref{eq:spspstar}), but when it came to calibrating the function $\alpha(t)$ to fulfill the third property (reproducing any arbitrarily given curve), we used the cumulative hazard rate. This was due to two main reasons. First, it made sense from a modeling perspective, as the model was designed to fit the market structure of survival probabilities or, equivalently, cumulative hazard rates at the initial time. The other reason was more practical: the link between survival probabilities and cumulative hazard rates linearized the problem, simplifying the calculations. Fortunately, going back to credit spreads is relatively easy. From equation \eqref{eq:spfinal} describing credit spreads, we have in RW:
\begin{equation*}
    \begin{aligned}
        \operatorname{Sp}^*(t, T) &=-\frac{1}{T-t} \ln \left[\delta+(1-\delta) \frac{S^m(0,T)}{S^m(0,t)}\frac{A(0,t)}{A(0,T)}\frac{e^{-B(0,t)y_0}}{e^{-B(0,T)y_0}}A(t,T)e^{-B(t,T)[\lambda^*(t)-\psi(t)]} \right]\\
        &=-\frac{1}{T-t} \ln \left[\delta+(1-\delta) S^*(t,T) \right]
    \end{aligned}
\end{equation*}
\begin{equation}
    \label{eq:spfuncLam}
    \hspace{-6.75cm} \operatorname{Sp}^*(t, T) =-\frac{1}{T-t} \ln \left[\delta+(1-\delta) e^{-\Lambda^*(t,T)} \right]
\end{equation}

This equation can also be inverted to: 
\begin{equation}
    \label{eq:LamfunSp}
    \Lambda^*(t,T) = -\operatorname{ln} \left[ \frac{1}{1-\delta} \left( e^{-(T-t)\operatorname{Sp}^*(t, T)}-\delta \right) \right]
    \footnote{This equation is not valid for any chosen $\operatorname{Sp}(,.)$ as it requires $\operatorname{Sp}(t,T) < -\frac{1}{T-t}\ln(\delta)$. This condition is also required in \cite{SarrCreditSpreads} to infer the initial survival probabilities from credit spreads. As stated by the authors, this condition is not restrictive. $\delta$ is the recovery rate and is often taken as 40\%, as we did in all our numerical simulations. For a $[t, T]$ gap as wide as 10, this condition allows for credit spreads up to 900 basis points (BPs).}.
\end{equation} 

\paragraph{$\bullet$  Practical use}\mbox{}\\
The few steps below summarize how the model can be used to make simulation under the real-world probability space. We recall that we have justified that the model under RW is arbitrage-free, and Section (\ref{sec:rwresults}), Figures \ref{fig:RW_2} and \ref{fig:RW_4} will show that we produce realistic curves. 

\begin{enumerate}
    \item First, one must choose the $N$ target values $\{c_i\}_{i \in \llbracket 1, N \rrbracket}$ for the cumulative hazard rate at dates  $\{t_i\}_{i \in \llbracket 1, N \rrbracket}$ and for the maturity $\widetilde{T}$:

    \begin{itemize}
        \item They can be chosen directly based on expert judgment.
        \item More likely, they will be obtained through targets from credit spreads and inverted to cumulative hazard rates using equation \eqref{eq:LamfunSp}.
    \end{itemize}

    \item Then we simulate $\Lambda(t,T)$ in RN (like any classic random simulation) using survival probabilities via equation \eqref{eq:survivprob} and $\Lambda(t,T)=-\operatorname{ln}[S(t,T)]$. 

    \item Using equation \eqref{eq:LamLamstar}, we infer $\Lambda^*(t,T)$. Note that this equation also requires the simulation of $y_t$.

    \item Finally, we obtain the RW credit spreads term structure using equation \eqref{eq:spfuncLam}.

\end{enumerate}

\section{Numerical tests: an application to economic forecasts and to stress tests}
\label{sec:rwresults}
Let us now take a look at a very practical use of RW credit spreads modeling. We introduce two real-world configurations to present two ways the model can be used. The first configuration is a forecast scenario, and the second is a stress scenario. 

In all the simulations that follow, we use the same data set and parameters as those used in \cite{SarrCreditSpreads}. In particular, the simulations are based on Credit Agricole's credit spread, survival probability, default intensity, and cumulative hazard rate. The initial time is taken as of January 1, 2024. The differentiation is made between its bond yield and the French Government's bond yield. The parameters are those calibrated in the so-called \textit{global scenario} (see Table \ref{table:resbase2} in Appendix \ref{sec:CIRParams}). Similarly, the diffusion of the cumulative hazard rate under RN follows the same logic as in the referenced paper. We start by diffusing the default intensity using its distribution, and using equation \eqref{eq:survivprob}, we get the cumulative hazard rate. We recall that a solution of the CIR model follows a noncentral chi-square, $\chi^2[2 c y(s) ; 2 q+2,2 w]$, with $2 q+2$ degrees of freedom and parameter of noncentrality $2 w$ (see \citep{cox2005theory}). The density of this distribution is given by:
\begin{equation*}
    f(y(t), t \; ; \; y(s), s)=c e^{-w-v}\left(\frac{v}{w}\right)^{q / 2} I_q\left(2(w v)^{1 / 2}\right), 
\end{equation*}
where $c := \frac{2 \kappa}{\sigma^2\left(1-e^{-\kappa(t-s)}\right)}$, $w := c y(s) e^{-\kappa(t-s)}$, $v := c y(t)$ and $q := \frac{2 \kappa \theta}{\sigma^2}-1$.

and $I_q(\cdot)$ is the modified Bessel function of the first kind of order $q$. To perform the simulations, we have used R’s library sde\footnote{\href{https://CRAN.R-project.org/package=sde}{https://CRAN.R-project.org/package=sde}}. Note that this distribution is conditional to the past values. Therefore, it requires setting a time step and keeping track of previous values. We take a weekly time step.

\paragraph{Goldman Sachs' 2024 Global Credit Outlook}\mbox{}\\
We first use the credit spreads forecasts of Goldman Sachs from their 2024 global credit report.\footnote{\href{https://www.goldmansachs.com/intelligence/pages/gs-research/2024-global-credit-outlook-back-in-the-saddle/report.pdf}{https://www.goldmansachs.com/intelligence/pages/gs-research/2024-global-credit-outlook-back-in-the-saddle/report.pdf}} The report anticipates a slight decrease in the values of credit spreads. The behavior expected for European financial issuers is described in Table \ref{tab:gs_base} below:

\begin{table}[H]
    \centering 
    \begin{tabular}{c|cccc}
        \hline \hline 
         & & & \\ [-0.4em] 
        \textbf{Current (Q4 2023)} & Q1 2024 & Q2 2024 & Q3 2024 & Q4 2024 \\ [0.4em] 
        \hline \hline 
        & & & \\
        201 & 194 & 190 & 187 & 184 \\ 
    \end{tabular}
    \caption{In BP, Goldman Sachs' Euro financial credit spreads forecast}
    \label{tab:gs_base}
\end{table}
These forecasts do not cover individual counterparties, including Crédit Agricole. However, we can easily transform them into relative changes and apply them to the behavior of Crédit Agricole's credit spreads. They are also not differentiated by maturity. We apply these forecasts to the 5-year tenor:
\begin{table}[H]
    \centering 
    \begin{tabular}{c|cccc}
        \hline \hline 
         & & & \\ [-0.4em] 
        \textbf{Current (Q4 2023)} & Q1 2024 & Q2 2024 & Q3 2024 & Q4 2024 \\ [0.4em] 
        \hline \hline 
        & & & \\
        113 & 109 & 107 & 105 & 103 \\ 
    \end{tabular}
    \caption{In basis points, Goldman Sachs' Euro financial credit spreads forecasts, translated to Crédit Agricole's 5-year tenor.}
    \label{table:ca_spstress}
\end{table}
Now, per the previous section, this stress should be transformed into a cumulative hazard rate stress. Applying equation \eqref{eq:LamfunSp}, it gives the following stress:

\begin{table}[H]
    \centering 
    \begin{tabular}{c|cccc}
        \hline \hline 
         & & & \\ [-0.4em] 
        \textbf{Current (Q4 2023)} & Q1 2024 & Q2 2024 & Q3 2024 & Q4 2024 \\ [0.4em] 
        \hline \hline 
        & & & \\
        0.096 & 0.093 & 0.091 & 0.089 & 0.088 \\ 
    \end{tabular}
    \caption{5Y Credit Agricole credit spreads forecasts translated to cumulative hazard rate forecasts}.
    \label{table:ca_hrstress}
\end{table}
For the simulation, we make 20,000 draws. The simulation starts on January 1, 2024, with the actual initial credit spreads. The time step is 1 week. We assume that the first target (Q1 2024) is attained after the first time step (January 8, 2024); then, between two quarters, the credit spread remains constant until the next forecast value, and so on. This means that we have one target $c_i$ for each time step $t_i$, and thus, the number of time steps $N$ is 52. In other terms we have for $I_0 = [01/01/2024, 1^{st}\operatorname{\text{ } week}[$, $I_1 = [1^{st}\operatorname{\text{ } week},\operatorname{Q1 \text{ } 2024}[$, $I_2 = [\operatorname{Q1 \text{ } 2024}, \operatorname{Q2 \text{ } 2024}[$, $I_3 = [\operatorname{Q2 \text{ } 2024}, \operatorname{Q3 \text{ } 2024}]$ and $I_4 = [\operatorname{Q3 \text{ } 2024}, \operatorname{Q4 \text{ } 2024}]$:        
\begin{equation*}
    \begin{aligned}
        \{t_i\}_{i \in \llbracket 1, 52 \rrbracket} &= \llbracket 1^{st}\operatorname{\text{ } week},\operatorname{Q1 \text{ } 2024}\rrbracket \cup \rrbracket\operatorname{Q1 \text{ } 2024}, \operatorname{Q2 \text{ } 2024}4\llbracket \cup \llbracket\operatorname{Q3 \text{ } 2024}, \operatorname{Q4 \text{ } 2024} \rrbracket \\
        c_i &= \begin{cases}
                0.096 \text{ if } t_i \in I_0\\
                0.093 \text{ if } t_i \in I_1\\
                0.091 \text{ if } t_i \in I_2\\
                0.089 \text{ if } t_i \in I_3\\
                0.088 \text{ if } t_i \in I_4\\
            \end{cases}
    \end{aligned}
\end{equation*}

Figures \ref{fig:RW_1} and \ref{fig:RW_2} below respectively show the behavior of the 5-year credit spread during the RW simulation and the behavior of the entire term structure.

 \begin{figure}[H]
    \centering
    \includegraphics[width=16.5cm]{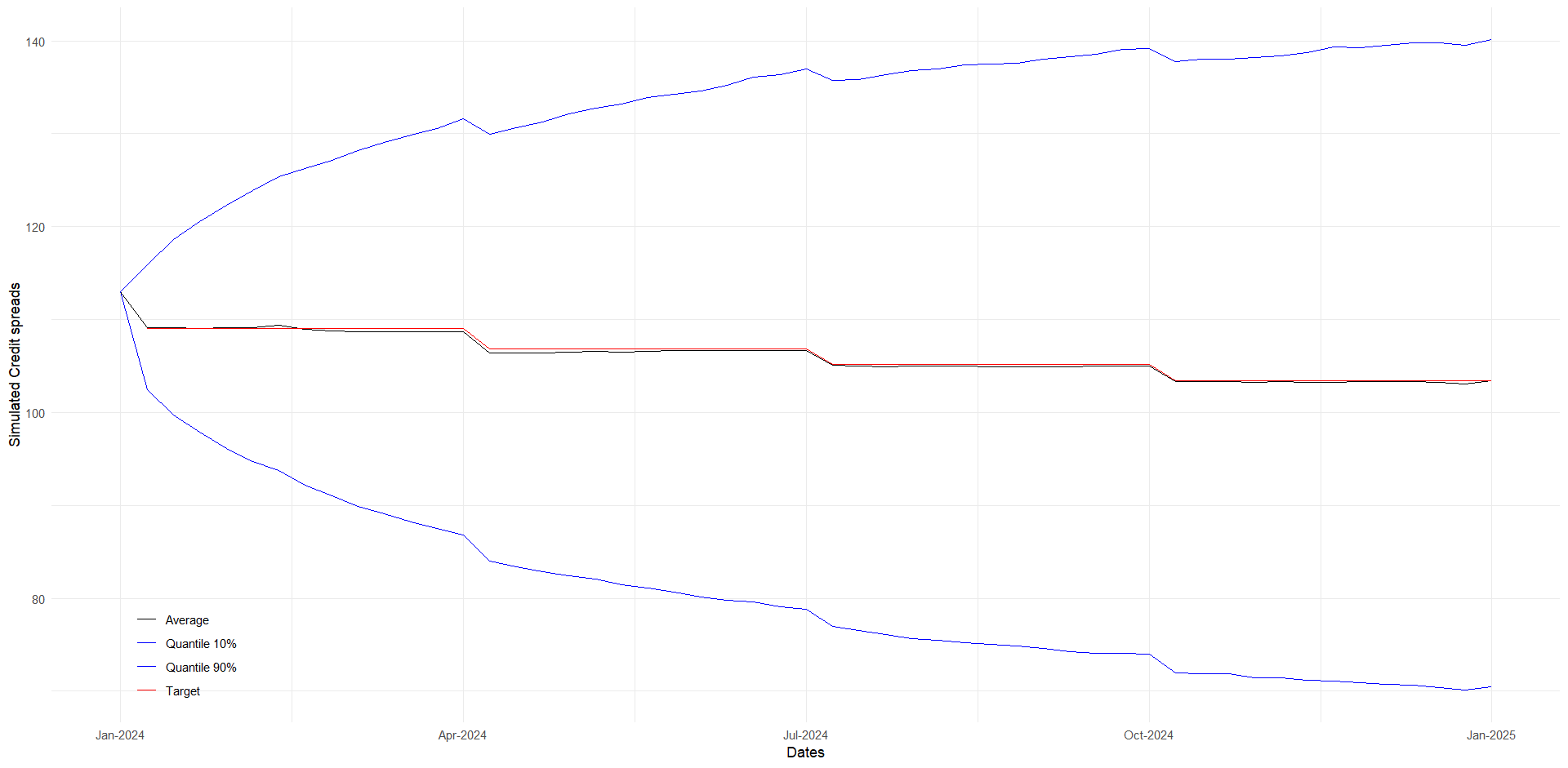}
    \caption{In basis points, 5-year Credit Agricole's credit spreads, simulated values and target values (Goldman Sachs's forecast). We observe the expectation and the 10th and 90th percentiles.}
    \label{fig:RW_1}
\end{figure}

We observe that the expectation of the simulated 5-year credit spreads (black curve) almost exactly matches the target values (red curve), while the quantiles (blue curves) encompass the expectation curve.

\begin{figure}[H]
    \centering
    \includegraphics[width=16.5cm]{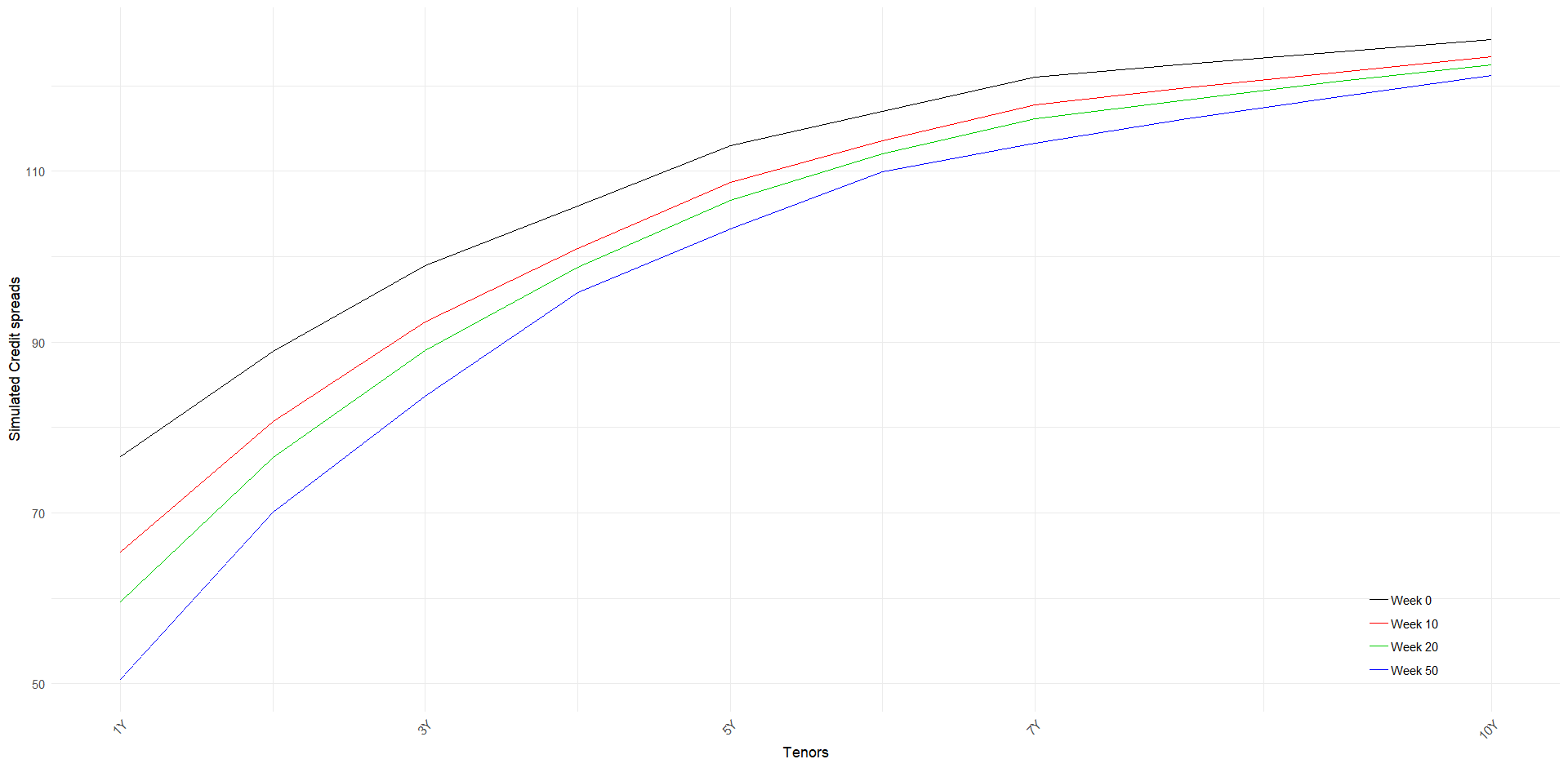}
    \caption{In basis points, the expectation of the term structure of Credit Agricole's credit spreads under the forecast scenario at different dates in the future.}
    \label{fig:RW_2}
\end{figure}

We observe how the entire term structure decreases as a consequence of the RW behavior of the credit spreads. This is indeed the intuitive result one would expect. Another very important observation is that the behavior imposed on the 5-year credit spread still results in very realistic credit spreads across the entire term structure.

\paragraph{EBA's 2023 EU-wide stress test}\mbox{}\\
The second configuration uses the European Banking Authority (EBA)'s regulatory stress test.\footnote{\href{https://www.eba.europa.eu/publications-and-media/press-releases/eba-launches-2023-eu-wide-stress-test}{https://www.eba.europa.eu/publications-and-media/press-releases/eba-launches-2023-eu-wide-stress-test}}. The stress test provides an absolute change to apply to issuers' credit spreads by region and sector, based on credit rating. For French financial counterparties rated between 1 and 2 by ECAIs (External Credit Assessment Institutions), which is the case of Credit Agricole as of January 1, 2024, the stress value is 133 BPs. While the regulator usually requires an instantaneous stress test, we consider here a linear progression of the stress value over a one-year period, for the sake of illustration. We still take the 5-year maturity. The simulation still involves 20,000 draws with a 1-week time step. Therefore, we have $N=52$ and $\forall i \in \llbracket 1, N \rrbracket$:
\begin{equation*}
    \begin{aligned}
        t_i &= i \frac{1}{52} \\
        c_i &= 113 + 133t_i
    \end{aligned}
\end{equation*}

Figures \ref{fig:RW_3} and \ref{fig:RW_4} below, respectively show the behavior of the 5-year credit spread during the RW simulation and the behavior of the entire term structure.

\begin{figure}[H]
    \centering
    \includegraphics[width=16.5cm]{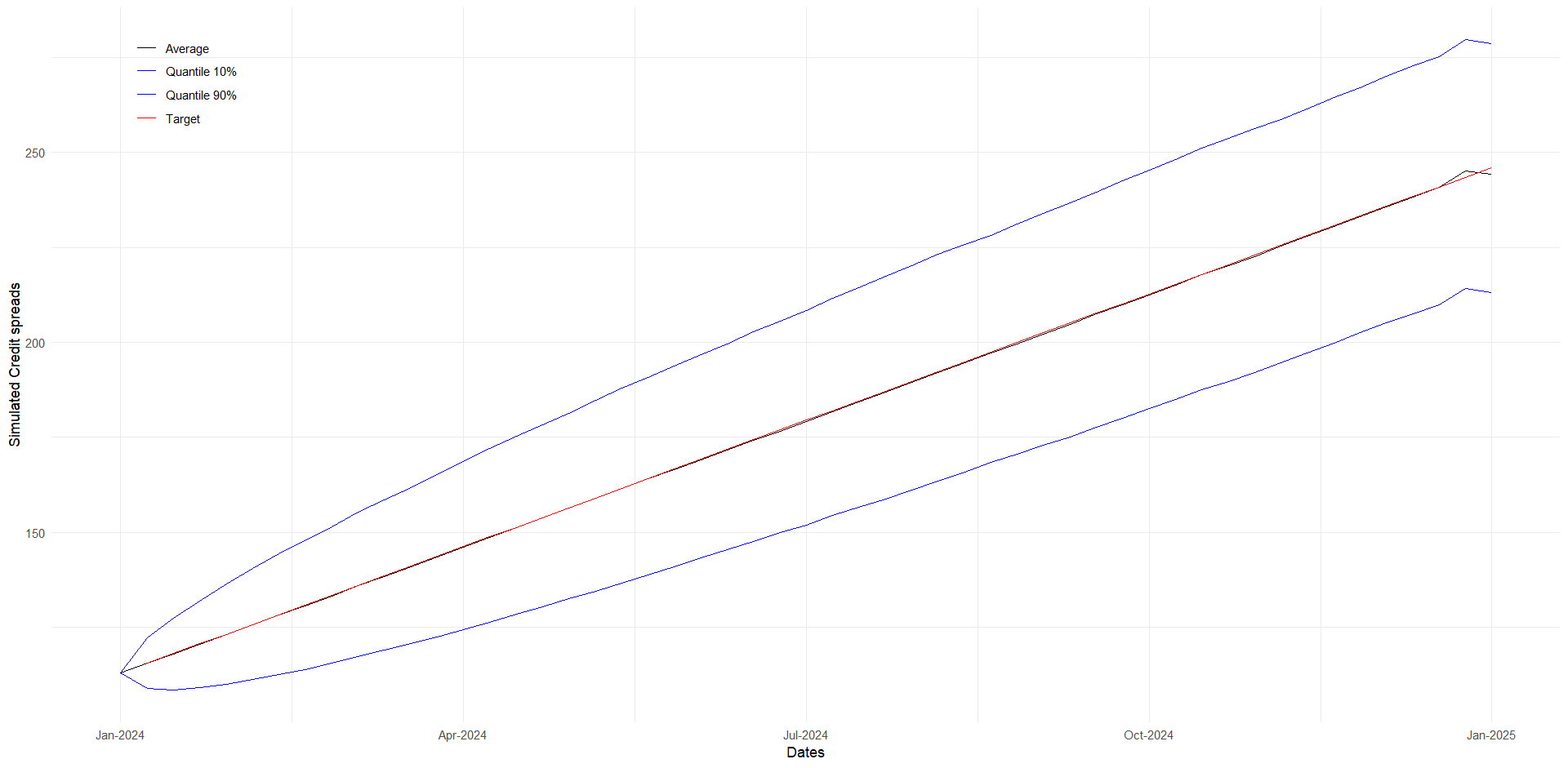}
    \caption{In basis points, 5-year Credit Agricole's credit spreads, simulated values and target values (EBA stressed values). We observe the expectation and the 10th and 90th percentiles.}
    \label{fig:RW_3}
\end{figure}

We observe here that, once again, we fit the target set for the 5-year maturity almost perfectly throughout the simulation. We include in Appendix \ref{sec:RWcumhaz} the graphics showing the behavior of the 5-year horizon cumulative hazard rate in both scenarios.

\begin{figure}[H]
    \centering
    \includegraphics[width=16.5cm]{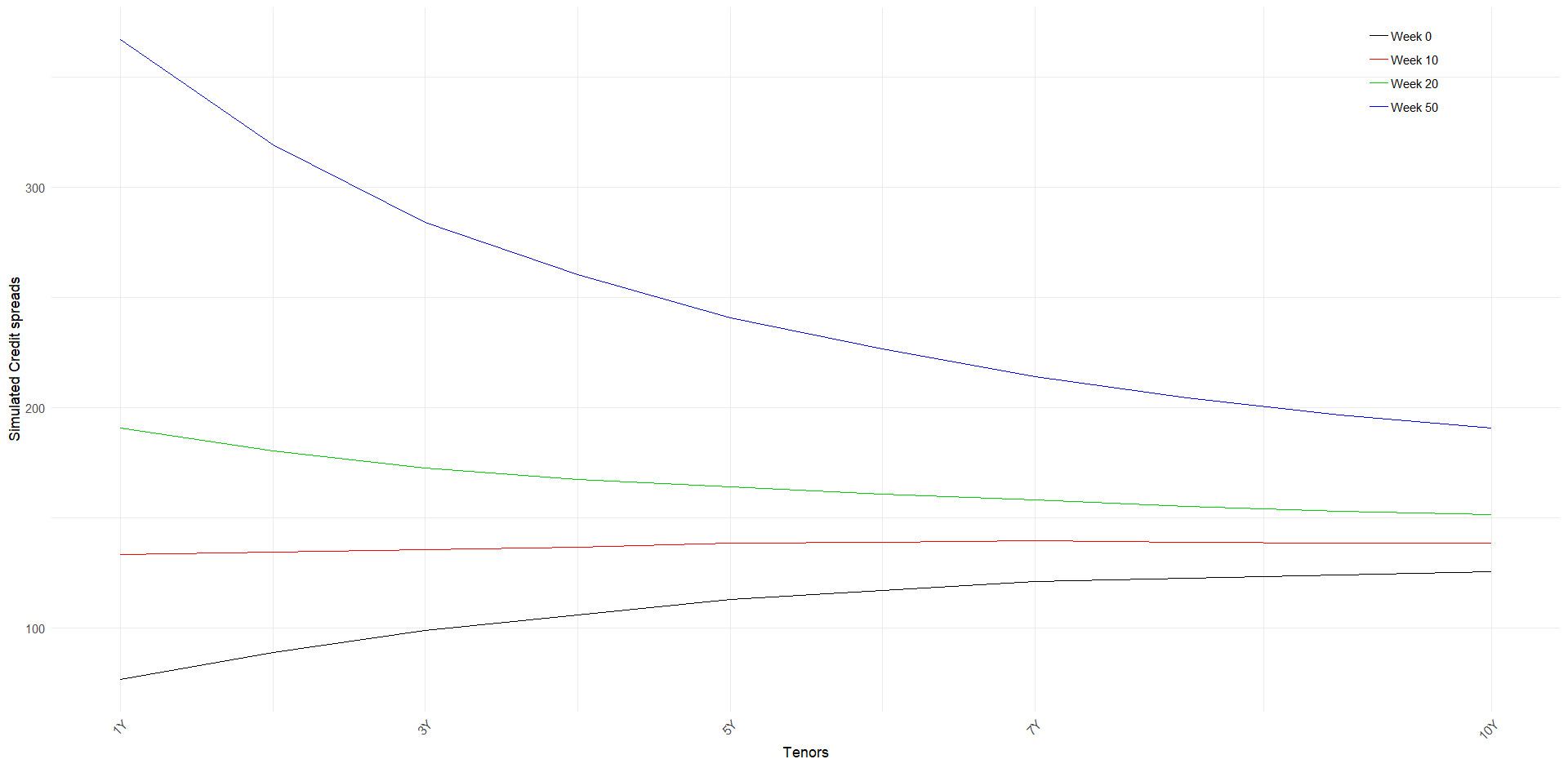}
    \caption{In basis points, the expectation of the term structure of Credit Agricole's credit spreads under the stress test scenario at different dates in the future.}
    \label{fig:RW_4}
\end{figure}

We can observe that, for this example, we have again, indeed produced a very realistic term structure of credit spreads. A notable feature is the inversion of the term structure as the stress increases. This behavior is observed in stressed market conditions, and we provide a few examples in Appendix \ref{sec:appendixInv}.

The process we have conducted in this section and the previous one can be generalized to other models. One needs to apply Theorem \ref{theo:rwrn} after verifying condition \ref{eq:Novikov}. Then, the one needs to find the equivalent of Theorem \ref{theo:linkLam} for their model. This will provide them with the equation to solve in order to calibrate $\alpha_t$. Note that this is straightforward for models with additive noise, as the function $\phi(x)$ is just the identity function.

\section{Conclusion}
In this research, we developed a comprehensive novel methodology for transitioning financial diffusion models from the risk-neutral measure to the real-world measure, leveraging probability theory results. Our framework has been applied to various models, including those with non-additive noise, such as the CIR++ model. The framework's utilization has been successfully demonstrated through numerical simulations on credit spreads' forecasts and stress tests. Drawing upon Credit Agricole's credit spread data relative to France's Government yield, we analyzed the robustness of our theoretical framework. Through applications to Goldman Sachs' 2024 global credit outlook forecasts and the European Banking Authority (EBA) 2023 stress tests, we validate the practical relevance and applicability of our model. The transition resulted in an arbitrage-free model capable of generating realistic credit spread term structures and accurately replicating arbitrarily chosen credit spread curves.

Our methodology offers practitioners and researchers a robust tool that enhances the precision and ease of conducting stress tests and forecasts. It allows for a detailed observation of the term structure of key risk indicators after applying forecasts or stress values to a specific maturity. Furthermore, the model accurately captures the behavior of the term structure in response to changes under high-stress conditions; the term structure of simulated RW credit spreads exhibits inversion of the curve, mirroring a behavior observed in the market. These modified term structures can be used for various risk calculations, including assessing the impact of interest rate variations on portfolio valuations through full revaluation of securities or derivatives. It can also be used for XVA calculations to estimate forward rates for determining future variable cash flows. The potential applications of this framework are extensive and varied, providing significant benefits for stress testing and risk management.

It is noteworthy that the framework does come with some complexities, particularly the need to verify conditions like the Novikov criterion, which can be tedious when dealing with models with non-additive noise. While our current model provides a solid foundation, future research avenues could explore enhancements such as incorporating jumps, generalizing the framework to dimensions greater than one, and applying it to other models and risk indicators, such as credit default swaps (CDS) premiums and foreign exchange (FX) rates. This would further expand the utility and applicability of the framework, making it a versatile tool in financial risk management and forecasting. Furthermore, the main advancement that could enhance the general applicability of the framework would be to generalize the approach by introducing a linearization for equation \eqref{eq:linear} and an error measurement for that linearization. For some models, the relation described by equation \eqref{eq:linear} might be very complex, making it difficult to express the indicator of interest under the real-world probability space as a function of the indicator under the risk-neutral probability space, as we do with Theorem \ref{theo:linkLam}.

\newpage 

\bibliographystyle{unsrtnat}
\bibliography{references}  

\begin{thebibliography}{23}
\providecommand{\natexlab}[1]{#1}
\providecommand{\url}[1]{\texttt{#1}}
\expandafter\ifx\csname urlstyle\endcsname\relax
  \providecommand{\doi}[1]{doi: #1}\else
  \providecommand{\doi}{doi: \begingroup \urlstyle{rm}\Url}\fi

\bibitem[Black and Scholes(1973)]{black1973pricing}
Fischer Black and Myron Scholes.
\newblock The pricing of options and corporate liabilities.
\newblock \emph{Journal of political economy}, 81\penalty0 (3):\penalty0
  637--654, 1973.

\bibitem[Berninger and Pfeiffer(2021)]{rw_g2++}
Christoph Berninger and Julian Pfeiffer.
\newblock The gauss2++ model: a comparison of different measure change
  specifications for a consistent risk neutral and real world calibration.
\newblock \emph{European Actuarial Journal}, 11\penalty0 (2):\penalty0
  677--705, 2021.

\bibitem[Hull et~al.(2014)Hull, Sokol, and White]{rw_hull}
John~C Hull, Alexander Sokol, and Alan White.
\newblock Modeling the short rate: the real and risk-neutral worlds.
\newblock \emph{Rotman School of Management Working Paper}, \penalty0
  (2403067), 2014.

\bibitem[Norman(2009)]{rw_bgm}
James~P Norman.
\newblock Real world interest rate modelling with the bgm model.
\newblock \emph{Available at SSRN 1480174}, 2009.

\bibitem[Brace et~al.(1997)Brace, Gatarek, and Musiela]{BGM_base}
Andrew Brace, Dariusz Gatarek, and Marek Musiela.
\newblock The market model of interest rate dynamics.
\newblock \emph{Mathematical Finance}, 7\penalty0 (2):\penalty0 127--155, 1997.
\newblock \doi{10.1111/1467-9965.00029}.

\bibitem[Cox et~al.(1985)Cox, Ingersoll, and Ross]{CIR_base}
John~C. Cox, Jonathan~E. Ingersoll, and Stephen~A. Ross.
\newblock A theory of the term structure of interest rates.
\newblock \emph{Econometrica}, 53\penalty0 (2):\penalty0 385--407, 1985.
\newblock ISSN 00129682, 14680262.
\newblock URL \url{http://www.jstor.org/stable/1911242}.

\bibitem[Dai and Singleton(2000)]{ConstantMarketPriceOfRisk}
Qiang Dai and Kenneth~J Singleton.
\newblock Specification analysis of affine term structure models.
\newblock \emph{The journal of finance}, 55\penalty0 (5):\penalty0 1943--1978,
  2000.

\bibitem[Duffee(2002)]{BadForecast}
Gregory~R Duffee.
\newblock Term premia and interest rate forecasts in affine models.
\newblock \emph{The Journal of Finance}, 57\penalty0 (1):\penalty0 405--443,
  2002.

\bibitem[Bruti-Liberati et~al.(2010)Bruti-Liberati, Nikitopoulos-Sklibosios,
  and Platen]{bruti2010rw}
Nicola Bruti-Liberati, Christina Nikitopoulos-Sklibosios, and Eckhard Platen.
\newblock Real-world jump-diffusion term structure models.
\newblock \emph{Quantitative Finance}, 10\penalty0 (1):\penalty0 23--37, 2010.

\bibitem[Barker et~al.(2016)Barker, Dickinson, and Lipton]{barker2016rw}
Russell Barker, Andrew~Samuel Dickinson, and Alex Lipton.
\newblock Simulation in the real world.
\newblock \emph{Available at SSRN 2831726}, 2016.

\bibitem[Alaya et~al.(2024)Alaya, Kebaier, and Sarr]{SarrCreditSpreads}
Mohamed~Ben Alaya, Ahmed Kebaier, and Djibril Sarr.
\newblock Credit spreads' term structure: Stochastic modeling with cir++
  intensity.
\newblock 2024.
\newblock URL \url{https://arxiv.org/abs/2409.09179}.

\bibitem[Lopes and V{\'a}zquez(2018)]{rw_lmm_neg}
Sara~Dutra Lopes and Carlos V{\'a}zquez.
\newblock Real-world scenarios with negative interest rates based on the libor
  market model.
\newblock \emph{Applied Mathematical Finance}, 25\penalty0 (5-6):\penalty0
  466--482, 2018.

\bibitem[Alfonsi(2013)]{Alf2013}
Aur\'{e}lien Alfonsi.
\newblock Strong order one convergence of a drift implicit {E}uler scheme:
  application to the {CIR} process.
\newblock \emph{Statist. Probab. Lett.}, 83\penalty0 (2):\penalty0 602--607,
  2013.
\newblock ISSN 0167-7152.
\newblock \doi{10.1016/j.spl.2012.10.034}.
\newblock URL \url{https://doi-org.extranet.enpc.fr/10.1016/j.spl.2012.10.034}.

\bibitem[Derouich and Kebaier(2022)]{BdeKeb24}
Mouna~Ben Derouich and Ahmed Kebaier.
\newblock The interpolated drift implicit euler scheme multilevel monte carlo
  method for pricing barrier options and applications to the cir and cev
  models.
\newblock 2022.

\bibitem[Karatzas and Shreve(1991)]{karatzas1991stochastic}
I~Karatzas and S~Shreve.
\newblock \emph{Stochastic calculus and Brownian motion}.
\newblock Springer-Verlag, New York, 1991.

\bibitem[Ben~Alaya and Kebaier(2012)]{BAK}
Mohamed Ben~Alaya and Ahmed Kebaier.
\newblock Parameter estimation for the square-root diffusions: ergodic and
  nonergodic cases.
\newblock \emph{Stochastic Models}, 28\penalty0 (4):\penalty0 609--634, 2012.

\bibitem[Alfonsi et~al.(2015)]{ALFONSI}
Aur{\'e}lien Alfonsi et~al.
\newblock \emph{Affine diffusions and related processes: simulation, theory and
  applications}.
\newblock Springer, 2015.

\bibitem[Widder(2015)]{Widder}
David~Vernon Widder.
\newblock \emph{Laplace transform (PMS-6)}, volume~61.
\newblock Princeton university press, 2015.

\bibitem[Abramowitz et~al.(1988)Abramowitz, Stegun, and Romer]{Handbook}
Milton Abramowitz, Irene~A Stegun, and Robert~H Romer.
\newblock Handbook of mathematical functions with formulas, graphs, and
  mathematical tables, 1988.

\bibitem[Bateman(1953)]{transcen}
Harry Bateman.
\newblock \emph{Higher transcendental functions [volumes i-iii]}, volume~1.
\newblock McGRAW-HILL book company, 1953.

\bibitem[Gunning and Rossi(2022)]{identity}
Robert~C Gunning and Hugo Rossi.
\newblock \emph{Analytic functions of several complex variables}, volume 368.
\newblock American Mathematical Society, 2022.

\bibitem[Harrison and Pliska(1981)]{arbitragefree}
J~Michael Harrison and Stanley~R Pliska.
\newblock Martingales and stochastic integrals in the theory of continuous
  trading.
\newblock \emph{Stochastic processes and their applications}, 11\penalty0
  (3):\penalty0 215--260, 1981.

\bibitem[Cox et~al.(2005)Cox, Ingersoll~Jr, and Ross]{cox2005theory}
John~C Cox, Jonathan~E Ingersoll~Jr, and Stephen~A Ross.
\newblock A theory of the term structure of interest rates.
\newblock In \emph{Theory of valuation}, pages 129--164. World Scientific,
  2005.

\end{thebibliography}

\newpage

\appendix

\section{Expression of $f(t)$}
\label{sec:ftproof}
We give here the proof of equation \eqref{eq:ft} which gives: 

\begin{equation*}
    f(t) = 
    \begin{dcases}
        e^{-\frac{1}{2}\kappa t} \sum_{i=1}^m\alpha^i\left( e^{\frac{1}{2}\kappa t_{i+1}} - e^{\frac{1}{2}\kappa t_{i}} \right) \text{ if } t\leq t_N \text{ and $m$ s.t. } t \in [0, t_m] \\
        e^{-\frac{1}{2}\kappa t} \sum_{i=1}^{N-1}\alpha^i\left( e^{\frac{1}{2}\kappa t_{i+1}} - e^{\frac{1}{2}\kappa t_{i}} \right) + e^{-\frac{1}{2}\kappa t}\alpha^N\left( e^{\frac{1}{2}\kappa t} - e^{\frac{1}{2}\kappa t_{N}} \right) \text{ if } t> t_N
    \end{dcases}
\end{equation*}

\begin{proof}
    \begin{itemize}
        \item $t > t_N$, let $m$ be taken, s.t. $t \in [0,t_m]$
            \begin{equation*}
                \begin{aligned}
                    f(t) &= \frac{1}{2}\kappa e^{-\frac{1}{2}\kappa t}\int_0^t \alpha(u)e^{\frac{1}{2}\kappa u}du \\
                    &=\frac{1}{2}\kappa e^{-\frac{1}{2}\kappa t} \int_0^t \sum_{i=1}^N \mathbb{1}_{u\in [t_i, t_{i+1}]}\alpha^i e^{\frac{1}{2}\kappa u} du \\
                    &=\frac{1}{2}\kappa e^{-\frac{1}{2}\kappa t} \sum_{i=1}^m \int_{t_i}^{t_{i+1}} \alpha^i e^{\frac{1}{2}\kappa u} du \\
                    f(t) &= e^{-\frac{1}{2}\kappa t} \sum_{i=1}^m\alpha^i\left( e^{\frac{1}{2}\kappa t_{i+1}} - e^{\frac{1}{2}\kappa t_{i}} \right)
                \end{aligned}
            \end{equation*}
    
        \item $t\leq t_N$, let $m$ be taken, s.t. $t \in [0,t_m]$
        \begin{equation*}
            \begin{aligned}
                f(t) &= \frac{1}{2}\kappa e^{-\frac{1}{2}\kappa t}\int_0^t \alpha(u)e^{\frac{1}{2}\kappa u}du \\
                &= \frac{1}{2}\kappa e^{-\frac{1}{2}\kappa t}\int_0^{t_N} \alpha(u)e^{\frac{1}{2}\kappa u}du + \frac{1}{2}\kappa e^{-\frac{1}{2}\kappa t}\int_{t_N}^t \alpha(u)e^{\frac{1}{2}\kappa u}du \\
                &= \frac{1}{2}\kappa e^{-\frac{1}{2}\kappa t}\int_0^{t_N} \sum_{i=1}^{N-1} \mathbb{1}_{u\in [t_i, t_{i+1}]}\alpha^i e^{\frac{1}{2}\kappa u} du + \frac{1}{2}\kappa e^{-\frac{1}{2}\kappa t}\int_{t_N}^t \alpha(u)e^{\frac{1}{2}\kappa u}du \\            
                &= \frac{1}{2}\kappa e^{-\frac{1}{2}\kappa t} \sum_{i=1}^{N-1}{ \int_{t_i}^{t_{i+1}}  \alpha^i e^{\frac{1}{2}\kappa u} du} + \frac{1}{2}\kappa e^{-\frac{1}{2}\kappa t}\int_{t_N}^t \alpha^N e^{\frac{1}{2}\kappa u}du \\ 
                f(t)&= e^{-\frac{1}{2}\kappa t} \sum_{i=1}^{N-1}\alpha^i\left( e^{\frac{1}{2}\kappa t_{i+1}} - e^{\frac{1}{2}\kappa t_{i}} \right) + e^{-\frac{1}{2}\kappa t}\alpha^N\left( e^{\frac{1}{2}\kappa t} - e^{\frac{1}{2}\kappa t_{N}} \right)
            \end{aligned}
        \end{equation*} 
    \end{itemize}
\end{proof}

\section{Parameters of the CIR++ Intensity model}
\label{sec:CIRParams}
We take the parameters that resulted from the global calibration described in \cite{SarrCreditSpreads}.

\begin{table}[H]
    \centering 
    \begin{tabular}{c c c c}
        \hline \hline 
         & & & \\ [-0.4em] 
        $\kappa$ & $\theta$ & $\sigma$ & $y_0$ \\ [0.4em] 
        \hline \hline 
        & & & \\
        $5.138 \times 10^{-1}$ & $1.497 \times 10^{-2}$ & $8.904 \times 10^{-2}$ & $4.348 \times 10^{-2}$ \\ [0.5em] 
        & & & \\
    \end{tabular}
    \caption{Parameters calibrated for the \textit{Global scenario}}
    \label{table:resbase2}
\end{table}

\section{5Y horizon cumulative hazard rate behaviour in the RW simulations}
\label{sec:RWcumhaz}
\vspace{-9.5mm}
\begin{figure}[H]
    \centering
    \begin{minipage}{0.495\textwidth}
        \begin{figure}[H]
            \centering 
            \includegraphics[width=8.5cm]{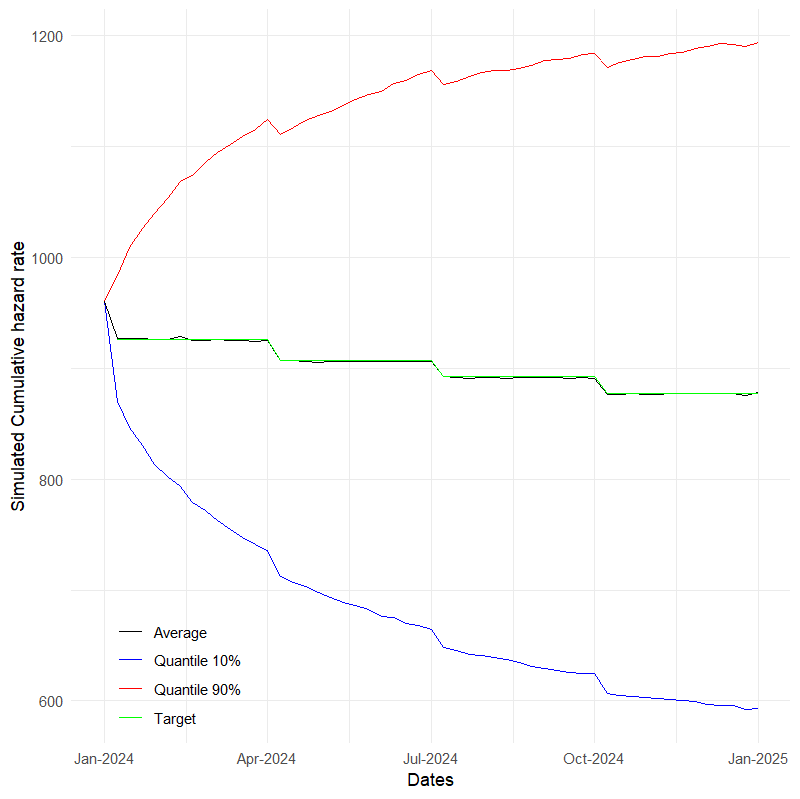} 
            \vspace{-1.5mm}
            \caption{In BP, 5Y Credit Agricole’s cumulative hazard rate under the forecasts for the 5Y scenario}
            \label{fig:avgdiff}
        \end{figure}
    \end{minipage}\hfill
    \begin{minipage}{0.495\textwidth}
        \begin{figure}[H]
            \centering
            \includegraphics[width=8.6cm]{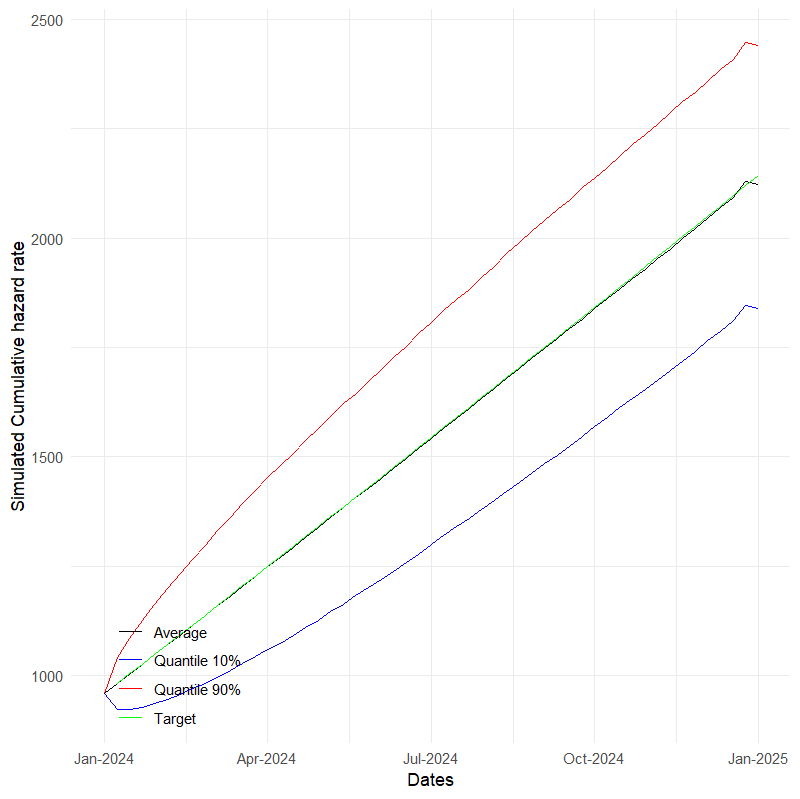} 
            \vspace{-1.7mm}
            \caption{In BP, 5Y Credit Agricole’s cumulative hazard rate and target values (EBA stressed values)}
            \label{fig:quantdiff}
        \end{figure}
    \end{minipage}
\end{figure}

\section{Inversion of the credit spreads curve}
\label{sec:appendixInv}
As mentioned in Section \ref{sec:rwresults}, we observed an inversion of the credit spread curve during periods of stress. Figure \ref{fig:invplot} shows inverted credit spreads during the stress period (dates within the Eurozone sovereign debt crisis in 2012 and 2013) compared to standard condition credit spreads (2022 and 2023).

\begin{figure}[H]
    \centering
    \includegraphics[width=16.5cm, height=7.5cm]{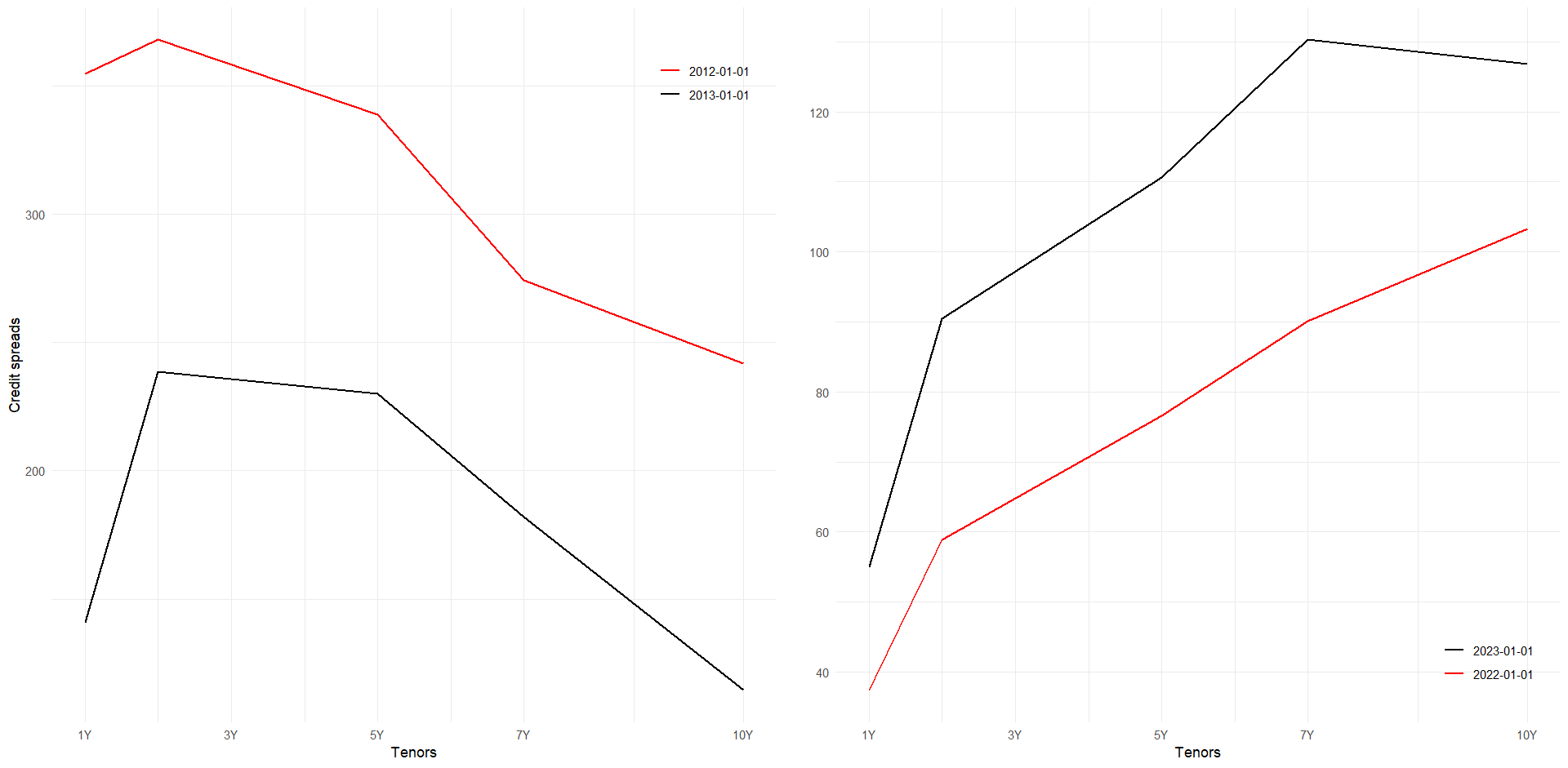}
    \caption{In basis points, Credit Agricole's credit spread term structure inverted during stress periods (2012/01/01 and 2013/01/01) and regular under standard conditions (2022/01/01 and 2023/01/01).}
    \label{fig:invplot}
\end{figure}

\end{document}